\PassOptionsToPackage{unicode}{hyperref}
\PassOptionsToPackage{hyphens}{url}
\PassOptionsToPackage{dvipsnames,svgnames,x11names}{xcolor}
\documentclass[
12pt]{article}

\usepackage{amsmath,amssymb}
\usepackage{iftex}
\ifPDFTeX
\usepackage[T1]{fontenc}
\usepackage[utf8]{inputenc}
\usepackage{textcomp} 
\else 
\usepackage{unicode-math}
\defaultfontfeatures{Scale=MatchLowercase}
\defaultfontfeatures[\rmfamily]{Ligatures=TeX,Scale=1}
\fi
\usepackage{lmodern}
\ifPDFTeX\else  
\fi
\IfFileExists{upquote.sty}{\usepackage{upquote}}{}
\IfFileExists{microtype.sty}{
	\usepackage[]{microtype}
	\UseMicrotypeSet[protrusion]{basicmath} 
}{}
\makeatletter
\@ifundefined{KOMAClassName}{
	\IfFileExists{parskip.sty}{%
		\usepackage{parskip}
	}{
		\setlength{\parindent}{0pt}
		\setlength{\parskip}{6pt plus 2pt minus 1pt}}
}{
	\KOMAoptions{parskip=half}}
\makeatother
\usepackage{xcolor}
\makeatletter
\ifx\paragraph\undefined\else
\let\oldparagraph\paragraph
\renewcommand{\paragraph}{
	\@ifstar
	\xxxParagraphStar
	\xxxParagraphNoStar
}
\newcommand{\xxxParagraphStar}[1]{\oldparagraph*{#1}\mbox{}}
\newcommand{\xxxParagraphNoStar}[1]{\oldparagraph{#1}\mbox{}}
\fi
\ifx\subparagraph\undefined\else
\let\oldsubparagraph\subparagraph
\renewcommand{\subparagraph}{
	\@ifstar
	\xxxSubParagraphStar
	\xxxSubParagraphNoStar
}
\newcommand{\xxxSubParagraphStar}[1]{\oldsubparagraph*{#1}\mbox{}}
\newcommand{\xxxSubParagraphNoStar}[1]{\oldsubparagraph{#1}\mbox{}}
\fi
\makeatother

\usepackage{longtable,booktabs,array}
\usepackage{calc} 
\usepackage{etoolbox}
\makeatletter
\patchcmd\longtable{\par}{\if@noskipsec\mbox{}\fi\par}{}{}
\makeatother
\IfFileExists{footnotehyper.sty}{\usepackage{footnotehyper}}{\usepackage{footnote}}
\makesavenoteenv{longtable}
\usepackage{graphicx}
\makeatletter
\def\maxwidth{\ifdim\Gin@nat@width>\linewidth\linewidth\else\Gin@nat@width\fi}
\def\maxheight{\ifdim\Gin@nat@height>\textheight\textheight\else\Gin@nat@height\fi}
\makeatother
\setkeys{Gin}{width=\maxwidth,height=\maxheight,keepaspectratio}
\makeatletter
\def\fps@figure{htbp}
\makeatother

\makeatletter
\@ifpackageloaded{caption}{}{\usepackage{caption}}
\AtBeginDocument{%
	\ifdefined\contentsname
	\renewcommand*\contentsname{Table of contents}
	\else
	\newcommand\contentsname{Table of contents}
	\fi
	\ifdefined\listfigurename
	\renewcommand*\listfigurename{List of Figures}
	\else
	\newcommand\listfigurename{List of Figures}
	\fi
	\ifdefined\listtablename
	\renewcommand*\listtablename{List of Tables}
	\else
	\newcommand\listtablename{List of Tables}
	\fi
	\ifdefined\figurename
	\renewcommand*\figurename{Figure}
	\else
	\newcommand\figurename{Figure}
	\fi
	\ifdefined\tablename
	\renewcommand*\tablename{Table}
	\else
	\newcommand\tablename{Table}
	\fi
}
\@ifpackageloaded{float}{}{\usepackage{float}}
\floatstyle{ruled}
\@ifundefined{c@chapter}{\newfloat{codelisting}{h}{lop}}{\newfloat{codelisting}{h}{lop}[chapter]}
\floatname{codelisting}{Listing}

\makeatother
\makeatletter
\@ifpackageloaded{caption}{}{\usepackage{caption}}
\@ifpackageloaded{subcaption}{}{\usepackage{subcaption}}
\makeatother

\ifLuaTeX
\usepackage{selnolig}  
\fi
\usepackage[]{natbib}
\usepackage{bookmark}

\IfFileExists{xurl.sty}{\usepackage{xurl}}{} 
\hypersetup{
	pdftitle={Title},
	pdfauthor={Author 1; Author 2},
	pdfkeywords={3 to 6 keywords, that do not appear in the title},
	colorlinks=true,
	linkcolor={blue},
	filecolor={Maroon},
	citecolor={Blue},
	urlcolor={Blue},
	pdfcreator={LaTeX via pandoc}}

\usepackage{mathtools}
\usepackage{amssymb}
\usepackage{amsmath}
\usepackage{bigints}
\usepackage{algpseudocode}
\usepackage{algorithm}
\usepackage{amsthm}
\usepackage{soul}
\usepackage{makecell}
\usepackage[capitalise]{cleveref}
\usepackage{subcaption}
\newtheorem{theorem}{Theorem}[section]

\newtheorem{lemma}{Lemma}[section]
\newtheorem{definition}{Definition}[section]
\newtheorem{assumption}{Assumption}[section]

\newtheorem{manualtheoreminner}{Theorem}
\newenvironment{manualtheorem}[1]{%
	\IfBlankTF{#1}
	{\renewcommand{\themanualtheoreminner}{\unskip}}
	{\renewcommand\themanualtheoreminner{#1}}%
	\manualtheoreminner
}{\endmanualtheoreminner}

\usepackage{tikz}
\usetikzlibrary{arrows.meta, shapes.multipart, positioning, calc}

\tikzset{
	treenode/.style={
		rectangle, draw=black, rounded corners, minimum width=16mm, minimum height=8mm,
		align=center, font=\footnotesize
	},
	term/.style={treenode, fill=gray!10},
	internal/.style={treenode, fill=blue!5},
	edge/.style={-Latex, line width=1pt},
	oplabel/.style={font=\small\bfseries}
}

\newcommand{\anon}{1}

\begin{document}

	\def\spacingset#1{\renewcommand{\baselinestretch}%
		{#1}\small\normalsize} \spacingset{1}


	\if1\anon
	{
		\title{\bf Conditional Copula models using loss-based Bayesian Additive Regression Trees}
		\author{Tathagata Basu\thanks{
				The authors gratefully acknowledge Leverhulme Trust for funding this research (Grant RPG-2022-026).}\hspace{.2cm}\\
			Newcastle University, UK\\
			Fabrizio Leisen \\
			Kings College London, UK\\
			Cristiano Villa \\
			Duke Kunshan University, China\\
			and \\
			Kevin Wilson \\
			Newcastle University, UK}
		\maketitle
	} \fi
	
	\if0\anon
	{
		\bigskip
		\bigskip
		\bigskip
		\begin{center}
			{\LARGE\bf Conditional Copula models using loss-based Bayesian Additive Regression Trees}
		\end{center}
		\medskip
	} \fi
	
	\bigskip
	\begin{abstract}
		The study of dependence between random variables under external influences is a challenging problem in multivariate analysis. We address this by proposing a novel semi-parametric approach for conditional copula models using Bayesian additive regression trees (BART) models. BART is becoming a popular approach in statistical modelling due to its simple ensemble type formulation complemented by its ability to provide inferential insights. Although BART allows us to model complex functional relationships, it tends to suffer from overfitting. In this article, we exploit a loss-based prior for the tree topology that is designed to reduce the tree complexity. In addition, we propose a novel adaptive Reversible Jump Markov Chain Monte Carlo algorithm that is ergodic in nature and requires very few assumptions allowing us to model complex and non-smooth likelihood functions with ease. Moreover, we show that our method can efficiently recover the true tree structure and approximate a complex conditional copula parameter, and that our adaptive routine can explore the true likelihood region under a sub-optimal proposal variance. Lastly, we provide case studies concerning the effect of gross domestic product on the dependence between the life expectancies and literacy rates of the male and female populations of different countries.
		
	\end{abstract}
	
	\noindent%
	\vfill
	
	\newpage
	\spacingset{1.75} 

	\section{Introduction}
	In multivariate analysis, we often observe that random variables exhibit a dependence in their behaviour, and the study of such dependence is extremely important due to its implications in modelling. However, analysing such dependence structures can be very cumbersome as we often require complex multivariate distributions, which makes this problem extremely challenging. Sklar's theorem\citep{sklar:1959} simplified such tasks, stating that a multivariate distribution can be modelled using it's marginals and a copula, that is, a joint cumulative distribution function on the unit hypercube. Several other theoretical advancements \citep{kimeldori1975uniform,ruschendorf1976,schweizer1981,Genest01091993} in the field have made copula inference popular in the context of statistical modelling, leading to its prominence in applied statistical literature involving  survival analysis \citep 	{clayton1978model,oakes1989bivariate,zheng1995estimates}; risk management \citep{engle1990asset,fama1993common}; engineering applications \citep{salvadori2007use,aghakouchak2010copula}; genetics \citep{li2006quantitative}, amongst others. While copula modelling has remained popular in multivariate analyses, a major computational improvement can be credited to \citet{bedford_vine_2002}, where they introduced a nested graphical model for copula construction and coined it as a `Vine Copula'. Later, based on their work \citet{czado_pair_cop_2010} introduced pair copula construction for multivariate modelling which improved the computational aspects further.  We refer readers to \citet{GENEST2024105278}'s review on copula modelling in remembrance of Abe Sklar for a detailed literature review.
	
	In many situations, the dependence between random variables can have an external influence, and in such cases, it is useful to adjust for external covariates in copulas. To tackle this issue, \citet{patton2006} formalised the conditional version of Sklar's theorem for application in financial time series modelling. Initial works on conditional copulas were mostly concerned with estimating time varying dependence structures using likelihood based approaches for autoregressive models \citep{patton2006,JONDEAU2006827,BARTRAM20071461}. Later \citet{acar2010} proposed a non-parametric approach for estimating conditional copulas for general problems. Several other likelihood-based non-parametric \citep{GIJBELS20111919} and semi-parametric \citep{ABEGAZ201243} approaches have been proposed in this regard. \citet{valle_cond_cop} proposed Dirichlet mixture models for estimating conditional copulas and \citet{GRAZIAN2022107417} proposed an approximate Bayesian approach. Recently, \citet{BonacinaLopezThomas+2025} proposed a ``Classification and Regression Tree'' (CART) based algorithm \citep{brei_CART} to model conditional copulas and investigated the consistency of the CART algorithm in conditional copula estimation.
	
	The introduction of the CART algorithm for conditional copula estimation motivates us to explore the viability of employing Bayesian additive Regression Trees (BART), the Bayesian alternative to CART introduced by \citet{chipman2010BART}. BART provides a generalisation of earlier Bayesian CART models \citep{chipman98BCART,denison98BCART}, where a single regression tree was used. Due to the flexible nature of BART, it has been explored in different contexts such as Poisson regression \citep{Murray03042021}; survival analysis \citep{Sparapani_BART}; gamma regression \citep{Linero_BART_gamma} and generalised BART \citep{Linero02012025}. Although the original tree prior proposed by \citet{chipman98BCART} is most commonly used in the BART literature, it is not very straightforward to incorporate prior information on the number of terminal nodes into this prior and the prior proposed by \citet{denison98BCART} tends to produce skewed trees. Several alternatives have been proposed \citep{Wu_CART,rockova_BART,Linero_BART_VS} to tackle this issue, but the choice of hyperparameters remains subjective. In order to reduce this subjectivity, \citet{serafini2024lossbasedpriortreetopologies} proposed a novel prior for regression trees using a loss based approach developed by \citet{villa_loss-prior}. The prior is formulated based on minimising the loss incurred due to the misspecification of the tree, which considers both the loss in information and complexity of the tree, making it appealing for BART models.
	
	In this article, we exploit the loss-based BART prior developed by \citet{serafini2024lossbasedpriortreetopologies} to introduce a new framework for modelling conditional copulas. However, unlike existing applications of BART models, we do not have a straightforward likelihood function to choose a conjugate prior. Consequently, we adopt a trans-dimensional Markov chain Monte Carlo (MCMC) algorithm that propose trees of varying dimensions, coupled with new terminal node values. Although a related method was recently proposed by \citet{Linero02012025}, their method relies on Laplace approximation to sample terminal node values. Unfortunately, such approximations are not viable for conditional copula modelling as the first and second order derivatives of the likelihood may not be numerically stable. To overcome this issue, we develop an efficient reversible jump MCMC (RJ-MCMC) type algorithm \citep{green_RJMCMC} tailored to the BART models. During the implementation, we notice that choosing the parameters of the proposal distribution is a rather difficult task for the sum of trees models and mixing can be extremely slow in some cases. So, we introduce an adaptive RJ-MCMC routine, that updates the variance of a Gaussian proposal by learning from previously sampled states. We further prove the ergodicity of this adaptive scheme, providing theoretical support for its applicability. Our initial analyses show that our both versions of our proposed RJ-MCMC algorithm perform well and are able to converge to the true likelihood region. Moreover, we notice that even for a subpar choice of an initial proposal variance, our adaptive RJ-MCMC routine can converge quickly to the high posterior region. Through a comprehensive simulation study with simulated conditional copulas, we illustrate our method's efficiency in identifying the true tree structure, as well as in estimating the true conditional dependence. Furthermore, due to the simple and explicable implementation of our adaptive RJ-MCMC algorithm, it can be employed to a wide range of other modelling problems showing its potential as a computation tool for BART models. 
	
	The rest of the paper is organised as follows: in \cref{sec:prelim} we introduce preliminary concepts of conditional copulas and Bayesian additive regression trees, followed by our semi-parametric estimation approach for conditional copulas in \cref{sec:cond:cop}. In \cref{sec:rjmcmc}, we discuss our proposed reversible jump MCMC routine for sampling from the posterior, along with its adaptive variant. After that, we illustrate our method using two synthetic datasets in \cref{sec:sim} to show its efficiency and accuracy; followed by case studies using the CIA world fact dataset in \cref{sec:cia}. Finally, we discuss the results shown in the paper and give conclusions in \cref{sec:conc}.
	
	\section{Preliminaries}\label{sec:prelim}
	
	In this section, we present a formal description of conditional copulas followed by the BART models and the loss-based prior for BART proposed by \citet{serafini2024lossbasedpriortreetopologies}, which we will incorporate in our modelling in later sections.
	
	\subsection{Conditional copula}
	In statistics, $C:[0,1]^d\to [0,1]$ is a $d$-dimensional copula if $C$ is a joint cumulative distribution function with uniform marginals. That is, for two uniform random variables $U_1$ and $U_2$ the copula is given by: $C(u_1,u_2) = P(U_1\le u_1, U_2\le u_2)$. For parametric families of copula, we denote it by $C(u_1,u_2\mid \theta$), where $\theta$ represents the parameter. Let $Y_1$ and $Y_2$ denote two random variables such that they have continuous marginal cumulative distribution functions (CDFs) $F_1(y_1)$ and $F_2(y_2)$. Then according to \citet{sklar:1959}, we can use a copula $C(\cdot\mid\theta)$ to model the dependence between $Y_1$ and $Y_2$ so that
	$H(y_1,y_2) = C\left(F_{1}(y_1),F_{2}(y_2)\mid \theta\right)$, where $H(y_1,y_2)$ is the joint CDF of $Y_1$ and $Y_2$. Moreover, since $F_i$'s are continuous marginal CDFs, we can get a random vector $(U_1,U_2)=(F_1(Y_1),F_2(Y_2))$ such that the marginals follow a uniform distribution. Now, let $u_i = F_{ix}(y_i)$ be the pseudo observations and $F_{ix}^{-1}(u_i)$s be the conditional quantile functions for $i=1,2$. Then
	\begin{equation*}
		C\left(u_1,u_2\mid \theta\right) = H\left(F_{1}^{-1}(u_1),F_{2}^{-1}(u_2)\right).
	\end{equation*}
	We also use this formulation to compute Kendall's $\tau$ in the following way \citep{nelsen2006}:
	\begin{equation*}
		\tau = 4\int \int C(u_1,u_2\mid \theta)dC(u_1,u_2 \mid \theta) - 1.
	\end{equation*}
	In many practical situations, the dependence between $Y_1$ and $Y_2$ can be influenced by an external variable $X$. To tackle such issues, \citet{patton2006} suggested a conditional version of \citet{sklar:1959}'s theorem. Let, $H_x(y_1,y_2)$ denote joint cumulative distribution of $Y_1$ and $Y_2$ conditional on $X$ and $F_{ix}(y_i)$ is the marginal CDF of $Y_i$ conditional on $X$ for $i=1,2$. Then we can write the following $H_x(y_1,y_2) = C\left(F_{1x}(y_1),F_{2x}(y_2)\mid \theta(x)\right)$, where $\theta(x)$ is the copula parameter as a function of $x$. Then the following two hold:
	\begin{align*}
		C\left(u_1,u_2\mid \theta(x)\right) = H\left(F_{1x}^{-1}(u_1),F_{2x}^{-1}(u_2)\right)
	\end{align*}
	and
	\begin{align*}
		\tau(x) = 4\int \int C(u_1,u_2\mid \theta(x)dC(u_1,u_2 \mid \theta(x) - 1
	\end{align*}
	where $F_{ix}$ is the conditional version of $F_{i}$ for $i=1,2$ and $\tau(x)$ is the conditional version of Kendall's $\tau$. For a more detailed introduction to the concepts, see \citet{patton2006,acar2010,GIJBELS20111919} etc.
	
	\subsection{Loss-based BART}
	
	Before defining the loss-based prior for BART, we first discuss regression tree models. Let $T$ denote a regression tree \citep{chipman98BCART} with a set of internal nodes depicting decision rules and a set of terminal nodes that assigns functional values corresponding to decision rules. We denote this set of terminal node values by $M =$ $\{\mu_1$,$\mu_2$, \dots, $\mu_{n_L(T)}\}$ where $n_L(T)$ is the number of terminal nodes of the tree $T$. Then we can model output $Z\coloneqq(Z_1,\cdots,Z_n)$ via:
	\begin{equation*}
		Z_i = g(x_i, T, M) + \epsilon_i; \qquad \epsilon_i\sim \mathcal{N}(0,\sigma^2)
	\end{equation*} 
	where $g$ is a function that assigns a terminal node value to $x_i$ for $i=1,\cdots,n$ and $\epsilon_i$ is the associated noise with variance $\sigma^2$. This is a single tree model developed by \citet{chipman98BCART}, which is extended to a sum of trees model by \citet{chipman2010BART} as:
	\begin{equation}\label{eq:BART}
		Z_i = \sum_{t=1}^m g(x_i, T_t, M_t) + \epsilon_i; \qquad \epsilon_i\sim \mathcal{N}(0,\sigma^2)
	\end{equation}
	where $T_t$ denotes the $t$-th tree; $M_t$ denotes the $t$-th vector of terminal node values $M_t =$ $\{\mu_1$,$\mu_2$, \dots, $\mu_{n_L(T_t)}\}$ of the $t$-th tree; and $n_L(T_t)$ denotes the number of terminal nodes of the $t$-th tree. This representation allows us to approximate a function with piecewise constant values on a partitioned domain, which are constructed by assigning a splitting rule $x_{\cdot j}\le \kappa$ for $1\le j \le p$ on each internal node, where the value $\kappa$ is either chosen from one of the observed values $x_{ij}$ or chosen uniformly within a range of values $(\underline{x}_{\cdot j},\overline{x}_{\cdot j})$.
	
	The most popular choice for a tree prior was introduced by \citet{chipman98BCART}, which was later used by \citet{chipman2010BART} for sum of regression trees models. Recently, \citet{serafini2024lossbasedpriortreetopologies} proposed an alternative prior for the tree topology. They consider a loss-based approach, which allows us to specify objective priors \citep{villa_loss-prior}. They suggested that the associated loss function for misspecification of a tree has two components: loss in information and loss in complexity. Further, they showed that the loss in complexity can be defined as $-\omega n_L(T_t) - \zeta \Delta(T_t)$, where $\omega\ge 0$ and $\zeta\in\mathbb{R}$ are prior parameters; $\Delta(T_t)$ is the difference between the number of right terminal nodes and left terminal nodes of $t$-th given tree and $n_L(T_t)$ is the number of terminal nodes of the $t$-th tree. So we consider $\pi(T_t) \propto \exp\left(-\omega n_L(T_t) - \zeta \Delta(T_t)\right)$ where $\pi$ denotes the density function.
	
	\section{Conditional Copula Modelling}\label{sec:cond:cop}
	
	Our main goal in this paper is to model the dependence structure of a conditional copula using regression trees. However, the conditional copula parameter $\theta(x)$ may have a specific range. Therefore, we use a suitable link function $h$ that maps a sum of trees model to the range of $\theta(x)$ so that $\theta(x_i) \coloneqq h\left(\sum_{t=1}^m g(x_i, T_t, M_t)\right)$ for $1\le i\le n$. 
	
	Now, let $c\left(u_1,u_2\mid \theta(x)\right)$ denote the conditional copula density function. Then, along with the loss-based prior for trees, we can define the following hierarchical model 
	\begin{align}\label{eq:bayes:hier}
		\begin{split}
			\pi\left(u_{1i},u_{2i} \mid \theta(x_i)\right) & = c\left(u_{1i},u_{2i}\mid h\left(\sum_{t=1}^m g(x_i, T_t, M_t)\right)\right)\\
			\pi(T_t) &\propto \exp\left(-\omega n_L(T_t) - \zeta \Delta(T_t)\right)\\
			\pi(M_t\mid T_t) &= \prod_{j=1}^{n_L(T_t)}\pi(\mu_j\mid T_t); \qquad t = 1,2,\cdots m.
		\end{split}
	\end{align}
	
	\paragraph*{Choice of link functions} The choice of link function to relate the tree structure with the conditional copula parameter is dependent on the family of the copula. Use of a link function is not unusual in the context of conditional copula modelling. \citet{ABEGAZ201243,valle_cond_cop} used link functions for calibrating the conditional copula parameter. In our case, we wish to keep the sum of trees model flexible so that it can take any value in $\mathbb{R}$. Therefore, we consider link functions that can map from $\mathbb{R}$ to the range of $\theta$ from a specific copula family, that we provide in \cref{tab:cop:link} for 5 different copula families.
	
	\paragraph*{Choice of priors on $\mu_j$}
	Earlier adaptions of BART models \citep{chipman2010BART,Sparapani_BART,Murray03042021} suggested a conjugate prior for the terminal node values $\mu_j$, as this allows us to compute the marginal likelihood easily. In our case, this is not possible so we consider we consider a default $\mathcal{N}(0,\sigma_{t}^2)$ on $\mu_j$ as suggested by \citet{chipman2010BART,Linero02012025} where $\sigma_{t}^2$ is a variance specific to the $t$-th tree. For $\sigma_{t}^2$ we use a flat inverse-gamma hyperprior.
	
	\paragraph*{Choice of $m$} The choice of the total number of trees remains an open problem. \citet{chipman2010BART} considered a default of 500 trees, whereas \citet{Linero02012025} suggested 200 trees for analysis. In our case, we follow a similar approach to that of \citet{serafini2024lossbasedpriortreetopologies}. We start with 5 trees and increase it by 5 at a time to monitor the change in performance for case studies. 
	
	\section{RJ-MCMC for Parameter Estimation}\label{sec:rjmcmc}
	
	In this section, we provide the MCMC algorithm for sampling from the posterior. In general, for BART models with a conjugate prior, a Metropolis Hastings (MH) type algorithm is used. In our case, we need to sample terminal node values alongside a new tree structure in each MCMC iteration. So we will employ a reversible jump MCMC type algorithm first proposed by \citet{green_RJMCMC}. Since, RJ-MCMC type algorithms are known for their slow mixing properties, we propose an adaptive alternative that improves the mixing speed, without needing to spend significant time on selecting the optimal parameter values of the proposal distribution.
	
	\subsection{Backfitting algorithm}
	
	The major building block of the sum of tree models is the backfitting algorithm. We modify this backfitting algorithm to meet our needs, which we will discuss below. First, let $(U_1,U_2)$ denote the random variable associated with the pseudo observations $\{(u_{1i},u_{2i}):1\le i \le n\}$ and $X$ denote the covariates. Then the posterior obtained from the hierarchical model described in \cref{eq:bayes:hier} is proportional to the following:
	\begin{align*}
		\begin{split}
			\prod_{i=1}^{n}c\left(u_{1i},u_{2i}\mid h\left(\sum_{t=1}^m g(x_i, T_t, M_t)\right)\right)\prod_{t=1}^{m}\pi(T_t)\prod_{t=1}^{m}\left(\prod_{j=1}^{n_L(T_t)}\pi(\mu_j\mid T_t)\right)\prod_{t=1}^{m}\pi(\sigma_{t}).
		\end{split}
	\end{align*}
	To sample from such posterior formulations, \citet{chipman2010BART} suggested a Bayesian backfitting algorithm where we sample the $k$-th pair $(T_k,M_k)$ conditional on the other $m-1$ pairs of $(T_t,M_t)$. To be more precise, let $(T_{-k},M_{-k})\coloneq \left\{(T_1,M_1),\cdots,(T_{k-1},M_{k-1}),(T_{k+1},M_{k+1}),\cdots(T_m,M_m)\right\}$,
	then in each MCMC iteration, we sample $(T_k, M_k)$  and $\sigma^2_k$ from the following conditional distributions: 
	\begin{equation*}
		(T_k,M_k)\mid T_{-k},M_{-k}, \sigma^2_{k}, U_1, U_2, X \qquad \text{and}\qquad \sigma^2_{k} \mid (T_k,M_k),U_1,U_2,X.
	\end{equation*}
	So, in order to implement the backfitting algorithm, we define 
	\begin{equation*}
		R_{ik} = \sum_{t\not=k}g(x_i, T_t, M_t)\quad\text{and}\quad R_{\cdot k}\coloneqq(R_{1k},R_{2k},\cdots,R_{nk}).
	\end{equation*}
	We will add these $R_{ik}$'s to the terminal node values of the $k$-th regression tree:
	\begin{align*}
		\pi(T_k,M_k \mid R_{\cdot k}, \sigma^2_{k}, U_1,U_2, X) &\propto \pi(T_k)\prod_{i=1}^{n}c\left(u_{1i},u_{2i}\mid h\left(R_{ik}+g(x_i, T_k, M_k)\right)\right)\prod_{j=1}^{n_L(T_k)}\pi(\mu_j\mid T_k)
	\end{align*}
	For the standard BART model, the use of a conjugate prior allows us to marginalise the likelihood with respect to $\mu\coloneqq\left(\mu_1,\mu_2,\cdots,\mu_{n_L(T_k)}\right)$. Unfortunately, we do not have the conjugacy property in our case, and we use a reversible jump MCMC algorithm \citep{green_RJMCMC} to sample from the posterior. We follow, a similar setup to that of \citet{Linero02012025} to build our algorithm without the Laplace approximation.
	
	\subsection{Proposal for RJ-MCMC}
	
	Reversible jump MCMC is a common strategy for trans-dimensional cases where we need to grow or reduce the dimension of our model. This allows us to work around the issue of not having a conjugate prior for the terminal node values and we can sample $(T_k,M_k)$ efficiently in each iteration. We consider a proposal function to generate a new pair $\left(T_k^{\ast}, M_k^{\ast}\right)$ at the $(\eta+1)$-th iteration given by:
	\begin{align}\label{eq:prop}
		q\left(T_k^{\eta},M_k^{\eta};T_k^{\ast}, M_k^{\ast}\right) = q\left(T_k^{\eta};T_k^{\ast}\right) q_{\left(T_k^{\eta};T_k^{\ast}\right)}\left(M_k^{\eta};M_k^{\ast}\right).
	\end{align}
	Here, $q\left(T_k^{\eta};T_k^{\ast}\right)$ denotes the tree proposal as described by \citet{chipman98BCART}. They suggested four different tree steps: \textsc{grow} move to randomly choose a terminal node and split it into two terminal nodes; \textsc{prune} move to randomly choose a parent of terminal nodes and turn it into a terminal node; \textsc{change} move to randomly choose an internal node and assign a new splitting rule and \textsc{swap} move to randomly choose a parent-child pair of internal nodes and swap their splitting rules. We provide a detailed description and construction in the appendix.
	
	The generation of new terminal node values relies on the proposed tree structure that we represent with $q_{\left(T_k^{\eta};T_k^{\ast}\right)}\left(M_k^{\eta};M_k^{\ast}\right)$. Clearly, we only need a proposal for \textsc{grow} and \textsc{prune} steps to get new candidates, as the dimension does not change for the other two tree steps. Therefore, for the \textsc{grow} move, we consider the $j$-th leaf to be grown to $j_l$ and $j_r$ so,
	$q_{\left(T_k^{\eta};T_k^{\ast}\right)}\left(M_k^{\eta};M_k^{\ast}\right)$ = $\pi_{prop}(\mu_{j_l})\times\pi_{prop}(\mu_{j_r})$ and for the \textsc{prune} move, we consider the $j_l$th and $j_r$th leaves to be pruned, then $q_{\left(T_k^{\eta};T_k^{\ast}\right)}\left(M_k^{\eta};M_k^{\ast}\right)$ = $\pi_{prop}(\mu_j)$
	where $\pi_{prop}$ denotes the proposal density function.
	
	For the choice of proposal we consider a normal distribution with mean being the terminal node values at the $\eta$-th iteration and variance $\sigma^2_{\text{prop}}$. We notice that choosing $\sigma^2_{\text{prop}}$ is rather difficult. Monitoring the likelihood gives us some idea but, due to slow mixing nature of RJ-MCMC, convergence may take a large number of iterations. So, we suggest an adaptive variance for the proposal. Our approach is motivated by the seminal work of \citet{haario_AMH} where they suggest the use of MCMC samples to update the covariance. They also provided a simple updating formula that reduced the computation cost of the covariance matrix. However, in our case, we need to adapt the method to facilitate the \textsc{grow} and \textsc{prune} moves. So we propose an adaptive covariance based on a partition of the predictor space. First, we define the following:
	\begin{definition}[Value at observation]\label{def:val:at:obs}
		Let $(T,M)$ be the regression tree along with the vector of terminal node values. Then the value at observation, $V_i$ is the terminal node value assigned to $x_i$. That is, $V_i = g(x_i, T,M)$.
	\end{definition}
	
	To perform our adaptive RJ-MCMC routine, we let our RJ-MCMC sampler run for $\eta_0$ iterations without any adaption. We collect the values at observations with respect to each tree and each iteration defined by: $V_{ik}^{\eta} = g\left(x_i,T_k^{\eta},M_k^{\eta}\right)$ for $1\le i \le n; 1\le \eta \le \eta_0$.
	
	Now, let $V_{.k}^{\eta}\coloneq (V_{1k}^{\eta}, \cdots, V_{nk}^{\eta})$ denote the vector of values at observation for the $k$-th tree at $\eta$-th iteration. We will use these vectors to compute the covariance matrix for all $\eta>\eta_0$ such that:
	\begin{equation*}
		C_k^{\eta} \coloneqq Cov\left(V_{\cdot k}^{1},V_{\cdot k}^{2},\cdots,V_{\cdot k}^{\eta}\right) + \epsilon\mathbf{I}_n \quad \text{for } \epsilon>0; \eta>\eta_0,
	\end{equation*}
	where $\mathbf{I}_n$ is the identity matrix of order $n$. This additional term $\epsilon\mathbf{I}_n$ ensures that the matrix is positive definite. Once we have this covariance matrix for the $t$-th tree, we can use this to calculate the proposal variance at a new terminal node value.
	
	Let $\mathbf{\Omega}^{k}\coloneqq \left\{\Omega_j^{k}\right\}_{j=1}^{n_L(T_k)}$ denote the partition of the observations created by the $k$-th tree. Then for any $\Omega_{j}^{k}\in \mathbf{\Omega}^{k}$, we can define a set of indices $\mathcal{I}_{j}^{k}\coloneqq \left\{i:x_i\in \Omega_{j}^{k}\right\}$ that we will use for calculating the proposal variance at the $j$-th terminal node of the $k$-th tree in the following way:
	\begin{equation}\label{eq:var:adapt}
		\sigma_{\text{prop};j}^2 \coloneqq \frac{2.4^2}{\left(\#\left\{\mathcal{I}_{j}^{k}\right\}\right)^3}
		\sum_{c\in\mathcal{I}_{j}^{k}}\sum_{d\in\mathcal{I}_{j}^{k}}\left[C_k^{\eta}\right]_{cd},
	\end{equation}
	where $\#(.)$ denotes the cardinality of a set. The scaling factor of $2.4$ is also used by \citet{haario_AMH} based on the optimal acceptance rate of Metropolis Hastings algorithms.
	
	Lastly, we update $C_k^{\eta+1}$ using the iterative formula for variances given by:
	\begin{equation*}
		C_k^{\eta+1} \coloneqq 
		\frac{\eta-1}{\eta} C_k^{\eta} 
		+ \frac{1}{\eta}\left(\eta 
		\left(\overline{V}_{\cdot k}^{\eta}\right)\left(\overline{V}_{\cdot k}^{\eta}\right)^T 
		- (\eta+1)\left(\overline{V}_{\cdot k}^{\eta+1}\right)\left(\overline{V}_{\cdot k}^{\eta+1}\right)^T 
		+ \left({V}_{\cdot k}^{\eta+1}\right)\left({V}_{\cdot k}^{\eta+1}\right)^T + \epsilon\mathbf{I}_n\right).
	\end{equation*}
	The formulation in \cref{eq:var:adapt} ensures that the proposal variance represents the variance of the mean value at observations at each terminal node, which is equivalent to using a marginal likelihood in traditional BART models. We summarise this algorithm in \cref{alg:ada:prop}.

	\begin{algorithm}[h]
		\caption{Computation of adaptive proposal}\label{alg:ada:prop}
		\begin{algorithmic}[1]
			\State Perform $\eta_0$ iterations with a fixed variance.
			
			\For{$k =1, \cdots, m$}
			\State Collect MCMC samples for initial $\eta_0$ iterations of the $k$-th tree:
			\begin{equation*}
				V_{ik}^{\eta} = g\left(x_i,T_k^{\eta},M_k^{\eta}\right)\quad 1\le i \le n; 1\le \eta \le \eta_0.
			\end{equation*}
			
			\State Calculate sample covariance matrix :
			\begin{equation*}
				C_k^{\eta} \coloneqq Cov\left(V_{\cdot k}^{1},V_{\cdot k}^{2},\cdots,V_{\cdot k}^{\eta}\right) + \epsilon\mathbf{I}_n \quad \text{for } \epsilon>0; \eta>\eta_0.
			\end{equation*}
			
			\State Calculate the proposal variance of the $j$-th terminal node value:
			\begin{equation*}
				\sigma_{\text{prop};j}^2 \coloneqq \frac{2.4^2}{\left(\#\left\{\mathcal{I}_{j}^{k}\right\}\right)^3}
				\sum_{c\in\mathcal{I}_{j}^{k}}\sum_{d\in\mathcal{I}_{j}^{k}}\left[C_k^{\eta}\right]_{cd};\qquad\mathcal{I}_{j}^{k}\coloneqq \left\{i:x_i\in \Omega_{j}^{k}\right\};\qquad 1\le j\le n_L(T_k).
			\end{equation*}
			
			\EndFor
			
			\State Update sample covariance using 
			\begin{equation*}
				C_k^{\eta+1} \coloneqq 
				\frac{\eta-1}{\eta} C_k^{\eta} 
				+ \frac{1}{\eta}\left(\eta 
				\left(\overline{V}_{\cdot k}^{\eta}\right)\left(\overline{V}_{\cdot k}^{\eta}\right)^T 
				- (\eta+1)\left(\overline{V}_{\cdot k}^{\eta+1}\right)\left(\overline{V}_{\cdot k}^{\eta+1}\right)^T 
				+ \left({V}_{\cdot k}^{\eta+1}\right)\left({V}_{\cdot k}^{\eta+1}\right)^T + \epsilon\mathbf{I}_n\right).
			\end{equation*}
		\end{algorithmic}
	\end{algorithm}
	
	Once we have the proposals for both trees and the terminal node values, we define the acceptance probability:
	\begin{equation}\label{eq:acc:prob}
		\alpha\left(T_k^{\eta},M_k^{\eta};T_k^{\ast}, M_k^{\ast}\right)
		= \min\left\{1,\frac{\mathcal{L}(U_1,U_2\mid T_k^{\ast},M_k^{\ast})\pi(T_k^{\ast},M_k^{\ast})q\left(T_k^{\ast}, M_k^{\ast};T_k^{\eta},M_k^{\eta}\right)}
		{\mathcal{L}(U_1,U_2\mid T_k^{\eta},M_k^{\eta})\pi(T_k^{\eta},M_k^{\eta}) q\left(T_k^{\eta},M_k^{\eta};T_k^{\ast}, M_k^{\ast}\right)}\right\},
	\end{equation}
	where $\mathcal{L}(U_1,U_2\mid T_k, M_{k})\coloneqq \prod_{i=1}^{n}c\left(u_{1i},u_{2i}\mid h\left(R_{ik}+g(x_i, T_k, M_k)\right)\right)$. Lastly, after sampling a new $(T_k^{\eta+1},M_k^{\eta+1})$, we update the values of $M_k^{\eta+1}$ using an MH step followed by a Metropolis within Gibbs step to update $\sigma_{k}^{2;(\eta+1)}$ conditional on $(T_k^{\eta+1},M_k^{\eta+1})$ using $\text{InvGamma}\left(a+\frac{n_L\left(T_k^{\eta+1}\right)}{2} , b + \frac{\sum_{j=1}^{n_L\left(T_k^{\eta+1}\right)}\mu_j^2}{2}\right)$ distribution. We summarise our sampling strategy in \cref{alg:MCMC}.
	
	\begin{algorithm}[h]
		\caption{$(\eta+1)$-th iteration MCMC sampling}\label{alg:MCMC}
		\begin{algorithmic}[1]
			\State Set $\theta(x_i) \leftarrow h\left(\sum_{t=1}^{m} g(x_i, T_t^{\eta},M_t^{\eta})\right)$ for $i = 1, \ldots, n$.
			\For{$k = 1, \ldots, m$}
			\State Set $R_{ik} \leftarrow \sum_{t\not=k}g(x_i, T_t^{\eta},M_t^{\eta})$ for $i = 1, \ldots, n$.
			
			\State Sample $(T_k^{\ast}, M_k^{\ast})$ by randomly choosing between the \textsc{grow}, \textsc{prune}, \textsc{change} and \textsc{swap}.
			
			\State Accept $\left(T_k^{\eta+1}, M_k^{\eta+1}\right)=\left(T_k^{\ast}, M_k^{\ast}\right)$ with probability $\alpha\left(T_k^{\eta},M_k^{\eta};T_k^{\ast}, M_k^{\ast}\right)$ or keep $\left(T_k^{\eta+1}, M_k^{\eta+1}\right)=\left(T_k^{\eta},M_k^{\eta}\right)$.
			
			\State Set $\theta(x_i) \leftarrow h\left(R_{ik} + g\left(x_i, T_k^{\eta+1}, M_k^{\eta+1}\right)\right)$ for $i = 1, \ldots, n$
			
			\State Update $M_k^{\eta+1}$ using an MH step targeting the full conditional.
			
			\State Sample $\sigma_{k}^{2;(\eta+1)}$ from $\text{InvGamma}\left(a+\frac{n_L\left(T_k^{\eta+1}\right)}{2} , b + \frac{\sum_{j=1}^{n_L\left(T_k^{\eta+1}\right)}\mu_j^2}{2}\right)$.
			\EndFor
		\end{algorithmic}
	\end{algorithm}

	\subsection{Convergence of the adaptive RJ-MCMC algorithm}
	
	In this section, we provide the theoretical justification of our proposed algorithm. We show that our adaptive RJ-MCMC algorithm converges. That is the Markov chain corresponding to the $k$-th tree, $\left\{\left(T_k^{\eta},M_k^{\eta}\right):\eta\in \mathbb{N}\right\}$ is ergodic. 
	
	For the sake of notational convenience, we drop the tree index $k$ as it plays no role. Now, let $\mathcal{S}\coloneqq\mathcal{T}\times \mathcal{M}$ be the state space for tree exploration equipped with a $\sigma$-algebra $\mathcal{F}$. Let,  $S^{\eta}\coloneqq\left(T^{\eta},M^{\eta}\right) \in \mathcal{S}$ be the $\eta$-th state of a tree and let $\Gamma^{\eta}\in \mathcal{Y}$ denote the random variable associated with the choice of kernel for updating from state $\eta$ to state $\eta+1$. Then we can define a family of kernels $\left\{\mathcal{K}_{\gamma}:\gamma\in\mathcal{Y}\right\}$ in the following way:
	\begin{equation*}
		P\left(S^{\eta+1} \in A\mid S^{\eta} = s, \Gamma^{\eta} = \gamma, \mathcal{G}^{\eta}\right) 
		= \mathcal{K}_{\gamma}(s;A)
		= \int_{A}\mathcal{K}_{\gamma}(s;ds^{\ast})
	\end{equation*}
	where $\mathcal{G}^{\eta}$ is a filtration generated by $\left\{(S^{\eta}, \Gamma^{\eta})\right\}$. That is, $\mathcal{G}^{\eta}$ an increasing sequence of $\sigma$-algebra such that $\mathcal{G}^{\eta} \coloneqq \sigma\left(S^1,\cdots,S^{\eta},\Gamma^1,\cdots,\Gamma^{\eta}\right)$.
	
	To ensure an ergodic RJ-MCMC algorithm, we require some assumptions/regularity conditions. We state these assumptions below.
	\begin{assumption}\label{ass:bounded:cop}
		$\exists$ \(\delta>0\) such that the copula density is bounded. That is, $0 < c(u_1,u_2\mid\theta(h(g(x,T,M)))) < \infty$ for all $(u_1,u_2)\in[\delta,1-\delta]^2$ and for all $M \in \mathcal{M}$
	\end{assumption}
	\begin{assumption}\label{ass:bounded:prior}
		All prior densities for the terminal node values are bounded above.
	\end{assumption}
	\begin{assumption}\label{ass:bounded:prop}
		All proposal densities of the terminal node values are uniformly bounded.
	\end{assumption}
	\begin{assumption}\label{ass:bounded:tree}
		With $n$ observations, the number of terminal nodes $n_L(T)\le n$ and for any tree with $n_L(T)\ge 2$, there exists at least one internal node that can be pruned.
	\end{assumption}
	
	\cref{ass:bounded:cop} is not very restrictive as in practice, we often use rank based pseudo observations for modelling copulas. \cref{ass:bounded:prior} is a direct consequence of choosing a prior density with finite variance and \cref{ass:bounded:prop} can be justified by the construction of the covariance matrix $C^{\eta}_k$. Since $C^{\eta}_k$ is uniformly bounded by construction \citep{haario_AMH}, the proposal variance obtained from $C^{\eta}_k$ are also uniformly bounded leading to a bounded proposal density. Lastly, \cref{ass:bounded:tree}, can be justified through the fact that at least one observation must lie in a terminal node and a binary regression tree always have an internal node with two terminal nodes. Now we state the main result regarding the ergodicity.
	
	\begin{theorem}
		If \cref{ass:bounded:cop}-\cref{ass:bounded:tree} are satisfied and $\gamma$ adaption strategy is defined by \cref{alg:ada:prop}. Then the following two conditions hold:
		\begin{enumerate}
			\item Diminishing adaptation $\rightarrow$ $\left\|\mathcal{K}_{\Gamma^{\eta+1}}(s;\cdot)-\mathcal{K}_{\Gamma^{\eta}}(s;\cdot)\right\|_{TV}\to 0$ in probability as $\eta\to \infty$.
			\item Containment $\rightarrow$ $\forall\iota>0$, $\exists N = N(\iota)\in \mathbb{N}$ such that 
			\begin{equation*}
				\left\|\mathcal{K}^N_{\gamma}(s;\cdot) - \pi(\cdot\mid u_1,u_2)\right\|_{TV} \le \iota
			\end{equation*}
		\end{enumerate}
		where $\pi(\cdot\mid u_1, u_2)$ denote the target distribution defined on $\mathcal{S}$.
	\end{theorem}
	
	\begin{proof}
		We provide the proof of the theorem in the appendix.
	\end{proof}
	
	The above two conditions ensure that an adaptive MCMC algorithm is ergodic as shown by \citet{roberts-rosenthal-ada_MCMC2007} and hence our proposed adaptive RJ-MCMC step is ergodic. Moreover, updating the terminal node values from full conditional satisfies the diminishing adaptation and is performed by a simple Metropolis Hastings step with Gaussian proposal. Therefore, it is equivalent to adaptive Metropolis algorithm proposed by \citet{haario_AMH}, which is proven to be ergodic. Lastly, updating the variance term $\sigma^2_k$ is ergodic due to properties of Metropolis within Gibbs sampling. Since all components of the backfitting algorithm are ergodic, our proposed MCMC algorithm for the BART based conditional copula models is also ergodic.

	\section{Simulation Studies}\label{sec:sim}
	
	In this section, we present our analyses using synthetic datasets simulated from two known data-generating processes, each replicated 100 times to evaluate the empirical performance of our approach. The first data-generating process is designed to assess the efficiency of our method in recovering the true model. To be precise, we examine whether the method can avoid overfitting and accurately identify the true tree structure underlying the data. The second data-generating process involves a more complex function and is used to evaluate the efficiency of our method in estimating the dependence structure for general functional relationships.
	
	For both data-generating processes, we investigate the performance of the proposed method for five different copula families: Gaussian, Students-t, Clayton, Gumbel and Frank. As we discussed in \cref{sec:cond:cop}, to model the conditional copula parameter, we consider different link functions, which we summarise in \cref{tab:cop:link}. This way, we ensure that the terminal node values can take any value in $\mathbb{R}$ and the link function maps the sum of trees model to the range of the conditional copula parameter. A similar approach was taken by \citet{ABEGAZ201243,valle_cond_cop} for calibrating the conditional copula parameter.
	
	\begin{table}[h]
		\centering
		\begin{tabular}{l|c|c|c|c}
			& Gaussian \& Student-t & Clayton & Gumbel & Frank \\
			\hline
			Support 
			& $\rho \in (-1,1)$ 
			& $\theta \in (0,\infty)$ 
			& $\theta \in [1,\infty)$ 
			& $\theta \in \mathbb{R}\setminus\{0\}$ \\
			Link 
			& $h(x)=\frac{\exp(x)-1}{\exp(x)+1}$ 
			& $h(x)=\exp(x)$ 
			& $h(x)=\exp(x)+1$ 
			& $h(x)=x$ \\
		\end{tabular}
		\caption{List of copula families used for our analyses followed by the range of the conditional copula parameter and the corresponding link function}
		\label{tab:cop:link}
	\end{table}
	
	For our analyses, we simulate 200 observations where the $x_i$-th predictor is sampled from a uniform distribution $U(0,1)$ for $1\le i\le 200$. The two sets of the conditional Kendall's taus are defined using the following equations:
	\begin{equation}\label{eq:tree:tau}
		\tau_1(x_i) = \begin{cases}
			0.3 & x_i \le 0.33\\
			0.8 & 0.33 < x_i \le 0.66\\
			0.3 & 0.66 < x_i
		\end{cases}\qquad\&\qquad \tau_2(x_i) = 0.2\sin(2\pi x_i) + 0.5.
	\end{equation}
	
	We use these $\tau_j(x_i)$'s ($j=1,2; 1\le i\le 200$) to simulate copula data using the relevant functions from \texttt{VineCopula} package in \texttt{R}. As mentioned earlier, we sample 100 replicates of each copula dataset to monitor the posterior coverage of our proposed method.
	
	For modelling the dependence structure, we consider the default loss-based prior ($\omega = 1.62$ and $\zeta = 0.62$) as suggested by \citet{serafini2024lossbasedpriortreetopologies}. For the hyperprior on $\sigma_t^2$, we consider InvGamma(1,2). We run 3000 MCMC iterations with 4 chains for all 100 replicates of each dataset. We employ both variants of our proposed RJ-MCMC algorithm; with and without adaption. For the adaptive version of the algorithm, we use the first 500 iterations for adaption. For convenience, from now on, we will use C-BART to denote our RJ-MCMC algorithm without adaption and A-C-BART for the adaptive version. For posterior estimates, we discard the first 1500 iterations.
	
	One specific goal of our analyses is to understand the efficiency of our method in capturing the true structure. We consider the posterior expectations averaged over 100 replicates. Additionally, we are interested in the predictive performance of our approach and we consider RMSE, 95\% credible interval length and 95\% credible interval coverage. 
	
	Let, $n_R, n_C, n$ denote the number of replicates, number of chains and number of observations respectively; and $\overline{\theta}_{ijk}, {\theta}^{(2.5)}_{ijk}, {\theta}^{(97.5)}_{ijk}$ denote the posterior mean, 2.5th posterior quantile and 97.5th posterior quantile of the conditional copula parameter at the $i$-th observation in the $j$-th MCMC chain of the $k$-th replicate. Let $\theta_i$ denote the true value of the conditional copula parameter at the $i$-th observation. Then RMSE, 95\% average credible interval length (CI-length) and 95\% credible interval coverage (CI-cov) are defined as follows
	\begin{equation*}
		\text{RMSE} = \frac{1}{n_R}\sum_{k=1}^{n_R}\left(\frac{1}{n_C}\frac{1}{n}\sum_{j=1}^{n_C}\sum_{i=1}^{n}\left(\theta_i - \overline{\theta}_{ijk}\right)^2\right),
	\end{equation*}
	\begin{equation*}
		\text{CI-length} = \frac{1}{n_R}\frac{1}{n_C}\frac{1}{n}\sum_{k=1}^{n_R}\sum_{j=1}^{n_C}\sum_{i=1}^{n}\left({\theta}^{(97.5)}_{ijk}-{\theta}^{(2.5)}_{ijk}\right),
	\end{equation*}
	and
	\begin{equation*}
		\text{CI-cov} = \frac{1}{n_R}\frac{1}{n_C}\frac{1}{n}\sum_{k=1}^{n_R}\sum_{j=1}^{n_C}\sum_{i=1}^{n}\mathbb{I}\left({\theta}^{(97.5)}_{ijk}\le\theta_i\le {\theta}^{(2.5)}_{ijk}\right)
	\end{equation*}
	where $\mathbb{I}(.)$ is the indicator function.
	
	\subsection{Example with true tree structure} 
	
	We use the first case ($\tau_1(x)$) to investigate the efficiency of our proposed approach in capturing the tree structure. For this initial analysis, we consider a single tree. This way, we can check how close our posterior estimates are to the true number of terminal nodes ($3$) and true depth ($2$). For the Gaussian, Student-t, Clayton and Gumbel copulas, we set our initial proposal variance to be equal to $0.2$, whereas for the Frank copula, we set our proposal variance to be equal to $1$. The main reason behind this is the type of link function we chose for these copulas. Since for the frank copula, the link function is the identity, the algorithm takes many more iterations to reach the true likelihood region. We also perform a complete analyses with proposal variance equal to 0.2, which we provide in the appendix. To assess the efficiency of our approach in recovering the true model, we consider the following quantities: the posterior expectations of the number of terminal nodes and depth averaged over the 100 replicates, along with their respective standard deviations over 100 replicates. We denote these quantities by `Mean $\hat{n}_L$', `SD $\hat{n}_L$', `Mean $\hat{D}$' and `SD $\hat{D}$'.
	
	We present the results in \cref{tab:eff:ex1}. We notice that for both C-BART and A-C-BART our posterior estimates of the number of terminal nodes and depths are close to their true values. The slight deviation can be attributed to the tree exploration steps where the model tends to grow leaves in tree MCMC steps. We also present the averaged acceptance rate over 100 replicates along with the standard deviation. We notice that the acceptance rate lies within $[0.16, 0.18]$, which is a reasonable range for a transdimensional MCMC algorithm. A similar effect is reported by \citet{Linero02012025} regarding the mixing of the likelihood. Additionally, we present the prediction accuracy and posterior coverage in \cref{tab:pred:ex1}. We notice that C-BART and A-C-BART are in agreement, though C-BART tends to perform better in terms of posterior coverage.
	
	Recall, that for the Frank copula family, we fixed the initial proposal variance to be equal to 1, while the proposal variance for all other families was fixed at 0.2. We perform an additional analysis, fixing the proposal variance at 0.2 for the Frank copula. This case, gives us an insight into the benefit of the adaptive version of our algorithm. We notice that the posterior estimates for the number of terminal nodes and depth of the tree are closer to the true value for A-C-BART and the credible interval coverage is significantly better than C-BART. 
	
	\begin{table}[h]
		\centering
		\begin{tabular}{l|l|cccccc}
			& & Mean $\hat{n}_L$ & SD $\hat{n}_L$ & Mean $\hat{D}$ & SD $\hat{D}$ & Mean Acc. & SD Acc. \\ 
			\hline
			C-BART & Gaussian & 3.149 & 0.110 & 2.059 & 0.095 & 0.161 & 0.011 \\ 
			& Student-t & 3.109 & 0.114 & 2.016 & 0.093 & 0.161 & 0.010 \\ 
			& Clayton & 3.177 & 0.065 & 2.046 & 0.055 & 0.179 & 0.007 \\ 
			& Gumbel & 3.143 & 0.044 & 2.033 & 0.036 & 0.171 & 0.010 \\ 
			& Frank & 3.308 & 0.167 & 2.190 & 0.142 & 0.174 & 0.011 \\
			\hline
			A-C-BART & Gaussian & 3.141 & 0.107 & 2.049 & 0.093 & 0.161 & 0.011 \\ 
			& Student-t & 3.084 & 0.150 & 1.995 & 0.134 & 0.160 & 0.011 \\ 
			& Clayton & 3.199 & 0.093 & 2.065 & 0.079 & 0.179 & 0.008 \\ 
			& Gumbel & 3.146 & 0.045 & 2.035 & 0.034 & 0.170 & 0.010 \\ 
			& Frank & 3.251 & 0.174 & 2.141 & 0.150 & 0.172 & 0.010 \\
		\end{tabular}
		\caption{Performance of our proposed methods for simulated datasets using a tree based conditional Kendall's tau ($\tau_1(x)$). We present the average and standard deviation of the posterior expectation of the number of terminal nodes; the posterior expectation of the depth of the tree; and the acceptance rate.}
		\label{tab:eff:ex1}
	\end{table}
	
	\begin{table}[ht]
		\centering
		\begin{tabular}{l|ccc|ccc}
			\multicolumn{1}{c|}{} &
			\multicolumn{3}{c|}{C-BART} &
			\multicolumn{3}{c}{A-C-BART} \\
			\hline
			& RMSE & CI-length & CI-cov & RMSE & CI-length & CI-cov \\ 
			\hline
			Gaussian & 0.075 & 0.218 & 0.932 & 0.076 & 0.219 & 0.921 \\ 
			Student-t & 0.093 & 0.277 & 0.930 & 0.098 & 0.274 & 0.908 \\ 
			Clayton & 0.067 & 0.192 & 0.925 & 0.067 & 0.192 & 0.923 \\ 
			Gumbel & 0.069 & 0.217 & 0.959 & 0.070 & 0.216 & 0.936 \\ 
			Frank & 0.083 & 0.251 & 0.930 & 0.083 & 0.254 & 0.925 
		\end{tabular}
		\caption{Prediction accuracy of our proposed methods for simulated datasets using a tree based conditional Kendall's tau ($\tau_1(x)$). We split our results into two general columns: left for C-BART, right for A-C-BART. We create subcolumns under each column to present root mean squared error (RMSE); average 95\% credible interval length (CI-length); and 95\% credible interval coverage (CI-cov).}
		\label{tab:pred:ex1}
	\end{table}
	
	\subsection{Example with a more general function}
	We use the second case ($\tau_2(x)$) to assess the prediction accuracy of our proposed approach in estimating the dependence structure, which in this case is highly non-linear with respect to the conditioning variable. For illustration purposes, we use 5 trees to model the conditional copula. To illustrate the prediction accuracy, we used RMSE, average CI length and CI coverage. For posterior analysis we set the proposal variance to be equal to 0.2 for all copula families including Frank.
	
	We present these results in \cref{tab:pred:ex2}. We notice that A-C-BART and C-BART are in good agreements in terms of prediction accuracy. However, unlike our results for the first case, we notice that A-C-BART indeed performs better in terms of prediction as well as posterior coverage. This happens as A-C-BART can effectively have proposal variances of different scales for different trees, which is not the case for C-BART. As a result, A-C-BART is more efficient, when we have multiple trees.
	
	\begin{table}[h]
		\centering
		\begin{tabular}{l|ccc|ccc}
			\multicolumn{1}{c|}{} &
			\multicolumn{3}{c|}{C-BART} &
			\multicolumn{3}{c}{A-C-BART} \\
			\hline
			& RMSE & CI-length & CI-cov & RMSE & CI-length & CI-cov \\ 
			\hline
			Gaussian & 0.074 & 0.318 & 0.961 & 0.070 & 0.321 & 0.968 \\ 
			Student-t & 0.082 & 0.377 & 0.978 & 0.082 & 0.381 & 0.981 \\ 
			Clayton & 0.073 & 0.304 & 0.960 & 0.070 & 0.306 & 0.966 \\ 
			Gumbel & 0.079 & 0.329 & 0.952 & 0.077 & 0.333 & 0.957 \\ 
			Frank & 0.070 & 0.278 & 0.926 & 0.068 & 0.284 & 0.940 \\  
		\end{tabular}
		\caption{Prediction accuracy of our proposed method for simulated datasets using a non-linear conditional Kendall's tau ($\tau_2(x)$). We split our results into two general columns: left for C-BART, right for A-C-BART. We create subcolumns under each column to present root mean squared error (RMSE); average 95\% credible interval length (CI-length); and 95\% credible interval coverage (CI-cov).}
		\label{tab:pred:ex2}
	\end{table}
	
	Additional figures and tables for our analyses with synthetic datasets are provided in the appendix.
	
	\section{CIA world factbook data}\label{sec:cia}
	We present two case studies with real dataset to discuss the applicability of our method. We use the CIA world factbook data. The data is originally produced by the CIA for policymakers in the USA and contains important insights on different countries. We are particularly interested in the `life expectancy' and the `literacy' rates of male and female populations in different countries. As a conditioning factor, we consider the economy of these countries represented by the per capita gross domestic product. 
	
	To fit the copula, we first get the pseudo observations by converting the original data into their ranks and scale them to ensure the values lie between 0 and 1. That is, for $n$ observations of a  pair $(y_{1i}, y_{2i})$; $1\le i\le n$, we define the pseudo observations in the following way $u_{ij} = \frac{r_{ij}}{n+1}$ for $1\le i \le n;\ j = 1,2$ where $r_{ij}$ is the rank of $y_{ij}$ among $(y_{1j}, y_{2j},\cdots, y_{nj})$ for $j=1,2$. For illustration purposes we run 4 parallel chains of 50000 MCMC iterations with 5 and 10 trees for both C-BART and A-C-BART. For A-C-BART, we use the first 500 iterations to calculate the proposal variance. For monitoring goodness-of-fit we consider two different two-sample tests to understand the goodness-of-fit of the simulated values from the fitted copulas: the Cramer test for multivariate data \citep{BARINGHAUS2004190}; and the Fasano-Franceschini test \citep{fasano-franceschini}, which is a generalised version of Kolmogorov–Smirnov test. These tests are designed to check if two sets of samples belong to the same distribution. The null hypothesis is rejected if the $p$-value is below $0.05$. Since these are permutation based tests, we simulate 100 replicates of copulas using estimated parameter and perform these tests with 1000 permutations.
	
	\subsection{Life Expectancy}
	We analyse the dependence between the female life expectancy and male life expectancy conditioned on the per capita GDP in log-scale. We use a total 221 observations of countries/territories for our analyses, out of which only 180 countries/territories have unique per capita GDP. We present the empirical distributions of these observations against the log GDP in \cref{fig:data:dist:LE} along with the pseudo observations of female and male life expectancies. We see that the female life expectancy lies within $[56.1,93.7]$ with the average being 77 years. For the male population, these numbers are slightly lower; the range being $[52.8, 86]$ and the average being 72 years. We notice that the life expectancy has a very strong tail dependence and the Kendall's tau is approximately equal to 0.83. 
	
	\begin{figure}
		\centering
		\includegraphics[width = 0.95\linewidth]{"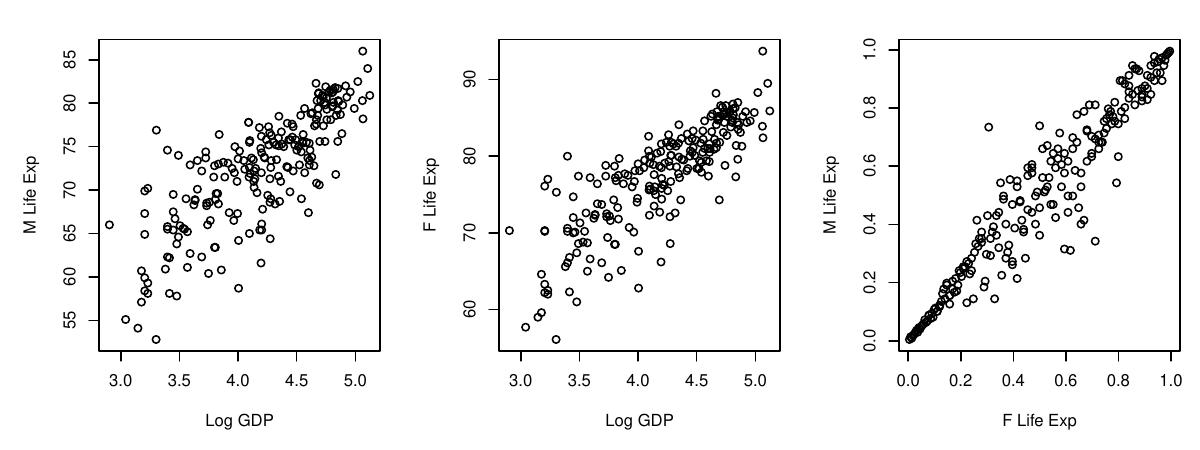"}
		\caption{Scatterplot of male life expectancy against log-GDP (left); female life expectancy against log-GDP (middle); and pseudo observations of female and male life expectancies (right).}
		\label{fig:data:dist:LE}
	\end{figure}
	
	We notice that the pseudo observations have an elliptical shape. Therefore, to model the dependence structure we consider Gaussian and Student-t copula families. We present the conditional taus in \cref{fig:taus:LE}, where the red (blue) solid line shows expected conditional taus for the Gaussian (Student-t) copula and the dashed lines show the credible interval. We notice that for both C-BART and A-C-BART the estimated conditional Kendall's taus are in good agreement for both families of copula except for one specific point, where the conditional Kendall's tau for Gaussian copula family gives us a very low value. This can be explained from \cref{fig:data:dist:LE} as well. We can see that for the same value of log-GDP there are outliers present in the scatterplots. We notice that countries with lower per capita GDP shows a very strong dependence between male and female life expectancies with the value being close to 0.9. As the per capita GDP increases, we notice that the dependence starts decreasing. For countries with extremely high per capita GDP, the dependence between male and female life expectancy stabilises near 0.8. This can also be verified from the scatterplots in \cref{fig:data:dist:LE}, where the shape of the empirical distributions of male and female life expectancies show strong resemblance for lower values of per capita GDP.
	
	\begin{figure}[h]
		\centering
		\includegraphics[width = 0.75\linewidth]{"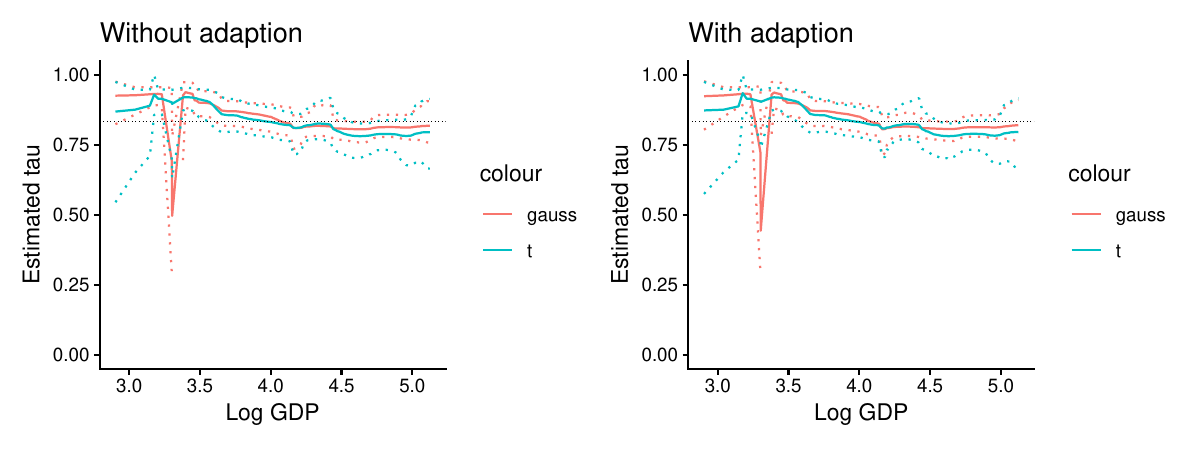"}
		\caption{Estimated dependence between male life expectancy and female life expectancy conditional on per capita log-GDP. For modelling, we use 5 trees and run 4 parallel chains of 50000 iterations. For posterior inference, we discard the first 5000 samples.}
		\label{fig:taus:LE}
	\end{figure}
	
	We present the traceplots of the log-likelihood in \cref{fig:trace:like:real:LE}. We notice that the likelihood stabilises in the same region for both C-BART and A-C-BART for the  Gaussian copula family. However, that is not the case for Student-t copula. We see that both C-BART and A-C-BART starts exploring higher likelihood regions for some chains. This is not completely unexpected as some of the countries have same per capita GDPs and that can lead to multimodality whilst estimating the conditional copula parameter.
	
	\begin{figure}[h]
		\centering
		\includegraphics[width = 0.75\linewidth]{"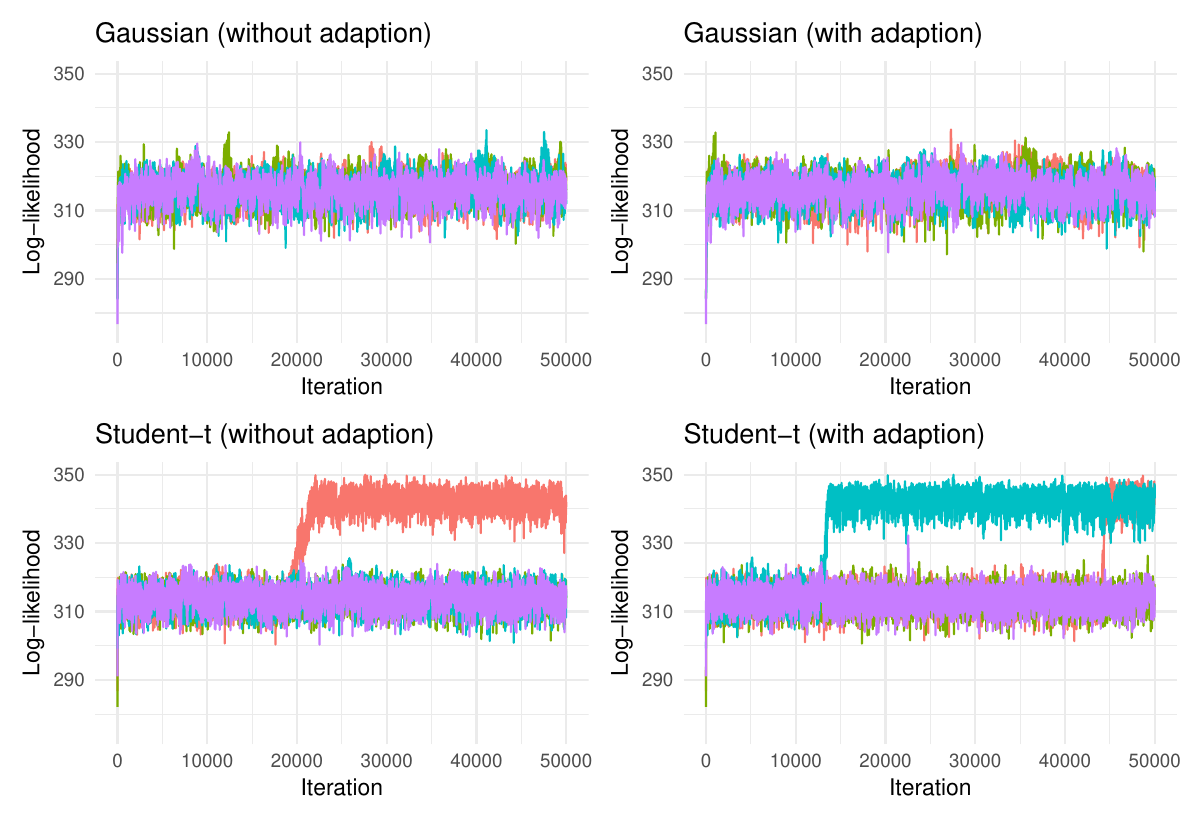"}
		\caption{Trace plots of the log-likelihood obtained from our analyses with life expectancies of the male and female populations. The plots are obtained by running 4 parallel chains, each with 50000 MCMC iterations and 5 trees. The left columns shows analyses with C-BART and the right column shows analyses with A-C-BART.}
		\label{fig:trace:like:real:LE}
	\end{figure}
	
	We present the results of goodness-of-fit tests in \cref{tab:LE:p-val}. We notice that the mean and median of $p$-values are well above 0.05 suggesting that there is no significant evidence that they are from two different distributions. Specifically, for the Student-t copula, the mean $p$-values are slightly higher when compared using the Fasano-Franceschini test. This suggets that Student-t copula is indeed a better choice for analysing the conditional dependence as the life expectancy of male and female populations exhibit a strong tail dependence. Moreover, we also present the goodness-of-fit test results for analyses with 10 trees. We notice no significant improvement in the goodness-of-fit test results. 
	
	We have provided additional figures for our analyses with 10 trees in the appendix. We notice that with 10 trees, our model tends to explore higher likelihood regions more frequently. However, it does not improve the goodness-of-fit. While it will be interesting to see whether the goodness-of-fit will improve after running more iterations of the MCMC chain, we refrain from doing so as we notice that the shape of the conditional Kendall's tau remains same. Albeit with more uncertainty accumulated.
	\begin{table}[h]
		\centering
		\begin{tabular}{l|l|ccc|ccc}
			\multicolumn{2}{c|}{} &
			\multicolumn{3}{c|}{Cramer test} &
			\multicolumn{3}{c}{FF test} \\
			\hline
			& & mean & median & sd & mean & median & sd \\ 
			\hline
			5 trees & Gaussian (C-BART) & 0.690 & 0.773 & 0.276 & 0.721 & 0.853 & 0.276 \\ 
			& Student-t (C-BART) & 0.694 & 0.777 & 0.272 & 0.735 & 0.853 & 0.270 \\ 
			& Gaussian (A-C-BART) & 0.689 & 0.774 & 0.276 & 0.721 & 0.844 & 0.278 \\ 
			& Student-t (A-C-BART) & 0.694 & 0.775 & 0.272 & 0.735 & 0.859 & 0.270 \\ 
			\hline
			10 trees & Gaussian (C-BART) & 0.689 & 0.774 & 0.276 & 0.710 & 0.828 & 0.283 \\ 
			& Student-t (C-BART) & 0.694 & 0.791 & 0.271 & 0.742 & 0.858 & 0.267 \\ 
			& Gaussian (A-C-BART) & 0.688 & 0.775 & 0.276 & 0.715 & 0.815 & 0.276 \\ 
			& Student-t (A-C-BART) & 0.694 & 0.793 & 0.271 & 0.739 & 0.849 & 0.270 \\ 
		\end{tabular}
		\caption{Goodness of fit tests of our proposed method for life expectancy conditional on log-GDP. We split our results into two general columns: left for the Cramer test and right for the Fasano-Franceschini test (FF test). We create subcolumns under each column to present the mean, the median and the standard deviation of the $p$-values obtained from 100 repetitions.}
		\label{tab:LE:p-val}
	\end{table}
	
	\subsection{Literacy}
	In this section, we discuss our analyses with the literacy rate of different countries. Unfortunately, the literacy rate is under-reported in the CIA world factbook dataset. We only have 167 observations with reported literacy rates and within this only 133 countries/territories have unique per capit GDP. We provide the scatterplots and distribution of pseudo observations for literacy rates of the male and female populations in \cref{fig:data:dist:LT}. For the female population, the literacy rate varies within $[18.2,100]$ and the average literacy rate is 84\%, whereas for the male population, the literacy rate lies in $[35.4,100]$ with the average literacy being 89\%. Unlike our example with the life expectancy data, the literacy rate does not have a linear trend against the log-GDP. However, Kendall's tau computed from the data is approximately 0.84, indicating a strong dependence between the literacy rates of the male and female populations. 
	\begin{figure}[h]
		\centering
		\includegraphics[width = 0.95\linewidth]{"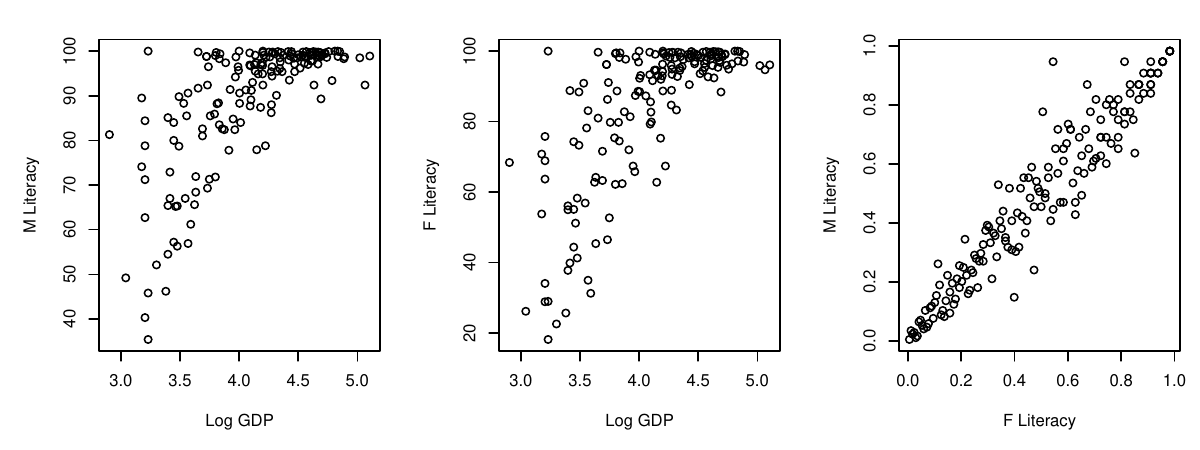"}
		\caption{Scatterplot of male literacy against log-GDP (left); scatterplot of female literacy against log-GDP (middle); and pseudo observations of female and male literacy rate (right).}
		\label{fig:data:dist:LT}
	\end{figure}
	
	Similar to our analyses with life expectancy, we notice that there is an overall agreement between the estimated conditional dependence using C-BART and A-C-BART. We notice that there is a strong similarity between the Gaussian and Student-t copula. We also see that the Gaussian copula exhibits an abrupt drop in the conditional dependence, which we present in \cref{fig:taus:LT}. Other than that, for both families, the literacy rates of the male and female populations have an almost constant dependence 
	with respect to the log-GDP.  We also present the traceplots of the log-likelihood in \cref{fig:trace:like:real:LT}. Unlike, our analyses with life expectancies, we notice that C-BART and A-C-BART tend to behave differently. For our analysis with the Gaussian copula, we notice that different chains remain at different likelihood regions. However, this is not the case for A-C-BART, where all the chains stabilises in the same likelihood region very quickly and two chains start exploring a higher likelihood region later. For our analysis with Student-t copula, the results are also very interesting. In this case, one of the chains of C-BART explores a higher likelihood region but exhibits a very high autocorrelation, whereas for A-C-BART, all the chains remain stable in the same region. 
	
	We present the results of goodness-of-fit tests in \cref{tab:LT:p-val}. Like before, we notice that there is no significant evidence that the simulated values from the fitted copulas and the pseudo observations generated from literacy rate data belong to two different families. Unlike our analyses with the life expectancy data, there is no clear winner. Additionally, we present our results with 10 trees and see no significant improvement.
	
	\begin{table}[ht]
		\centering
		\begin{tabular}{l|l|ccc|ccc}
			\multicolumn{2}{c|}{} &
			\multicolumn{3}{c|}{Cramer test} &
			\multicolumn{3}{c}{FF test} \\
			\hline
			& & mean & median & sd & mean & median & sd \\ 
			\hline
			5 trees & Gaussian (C-BART) & 0.691 & 0.775 & 0.272 & 0.708 & 0.782 & 0.249 \\ 
			& Student-t (C-BART) & 0.690 & 0.758 & 0.269 & 0.713 & 0.784 & 0.243 \\ 
			& Gaussian (A-C-BART) & 0.691 & 0.778 & 0.273 & 0.708 & 0.774 & 0.252 \\ 
			& Student-t (A-C-BART) & 0.690 & 0.761 & 0.269 & 0.711 & 0.785 & 0.244 \\ 
			\hline
			10 trees & Gaussian (C-BART) & 0.690 & 0.782 & 0.272 & 0.721 & 0.794 & 0.253 \\ 
			& Student-t (C-BART) & 0.689 & 0.757 & 0.268 & 0.708 & 0.750 & 0.241 \\ 
			& Gaussian (A-C-BART) & 0.691 & 0.786 & 0.272 & 0.721 & 0.778 & 0.251 \\ 
			& Student-t (A-C-BART) & 0.689 & 0.759 & 0.268 & 0.709 & 0.756 & 0.244 \\ 
		\end{tabular}
		\caption{Goodness of fit tests of our proposed method for literacy conditional on log-GDP. We split our results into two general columns: left for the Cramer test and right for the Fasano-Franceschini test (FF test). We create subcolumns under each column to present the mean, the median and the standard deviation of the $p$-values obtained from 100 repetitions.}
		\label{tab:LT:p-val}
	\end{table}
	\begin{figure}[h]
		\centering
		\includegraphics[width = 0.75\linewidth]{"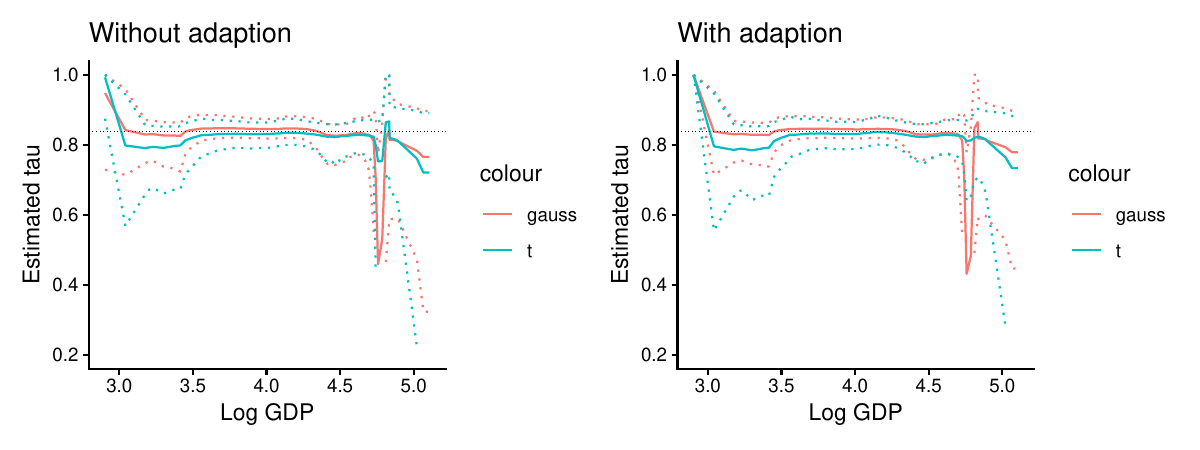"}
		\caption{Estimated dependence between male literacy and female literacy, conditional on log-GDP. For modelling, we use 5 trees and run 4 parallel chains of 50000 iterations. For posterior inference, we discard the first 5000 samples.}
		\label{fig:taus:LT}
	\end{figure}
	\begin{figure}[h]
		\centering
		\includegraphics[width = 0.75\linewidth]{"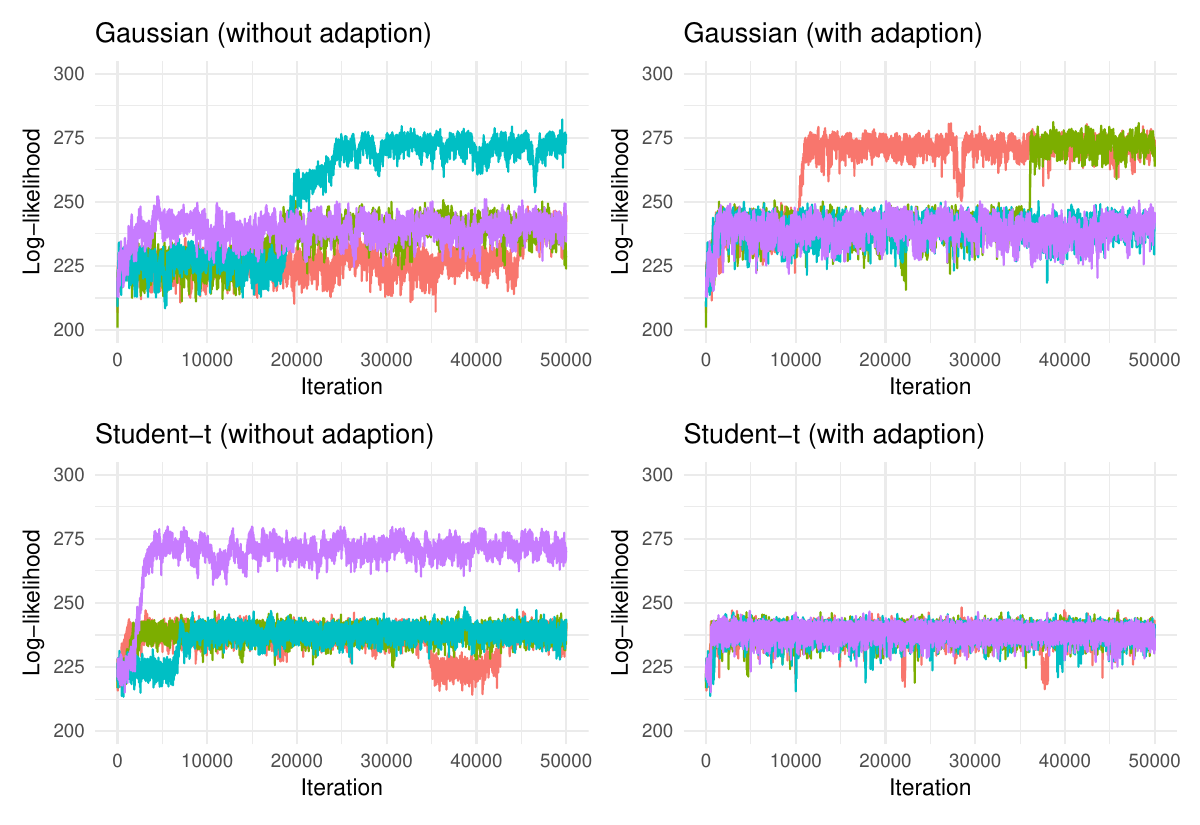"}
		\caption{Trace plots of the log-likelihood obtained from our analyses with literacy rates of the male and female populations. The plots are obtained by running 4 parallel chains, each with 50000 MCMC iterations and 5 trees. The left columns shows analyses with C-BART and the right column shows analyses with A-C-BART.}
		\label{fig:trace:like:real:LT}
	\end{figure}
	
	\section{Conclusion}\label{sec:conc}
	
	In this paper, we propose a Bayesian semi-parametric approach for modelling conditional copulas where we use the loss-based BART prior proposed by \citet{serafini2024lossbasedpriortreetopologies} to estimate the conditional copula parameter. To sample from the posterior, we propose a novel adaptive RJ-MCMC routine that can efficiently update the proposal variance of the terminal node values avoiding any manual tuning. In this way, we achieve a more general sampling algorithm for BART based models where we do not require a conjugate prior for posterior inference or a smooth likelihood function for applying other approximation based algorithms. Even though our proposed method is developed for copula modelling, the lack of restrictions in our model, makes it desirable for other areas of modelling where we can employ BART. In our framework, we use a link function to estimate the conditional copula parameter. However, in some cases, one might be interested in estimating the conditional Kendall's tau directly. Due to the flexible nature of our model, one can easily adapt our model to meet such needs by changing the link functions on the sum of trees model accompanied by a transformed beta prior \citep{gokhale_prior_cor} for the terminal node values.
	
	We perform extensive simulation studies to show the efficiency of our methods in recovering the true tree structure present within the data as well as their ability to approximate a complicated function efficiently. We notice that, specifically for complicated functions, where we require multiple trees, the adaptive routine performs better in reducing the prediction error, suggesting applicability. We also illustrate our method's adaptive capability using a case where our adaptive RJ-MCMC routine converges towards the true likelihood region even with a sub optimal choice for proposal variance, suggesting its usefulness in statistical analysis of real datasets where assigning proposal variance can be difficult. 
	
	We also present case studies with real dataset involving life expectancy and literacy rate of male and female populations of different countries. For conditioning, we consider per capita GDP of each country to see whether there is a dependence between male and female populations. We notice that our proposed method gives us some interesting results and empirical `goodness of fit' tests suggest that our method is capable of estimating the dependence structure correctly. Moreover, some of these case studies show the benefit of adaptation, as we notice that our adaptive routine tends to be more stable than the non-adaptive counterpart.
	
	While we get promising results with our method across simulation and real case studies, we still require an objective way to fix the number of trees in our BART based copula models. Currently, in our work we use the goodness-of-fit to see the improvement in models with respect to increasing number of trees. This gives us an empirical idea whether to use more trees or not. However, this makes our analyses computationally expensive as we need to fit multiple models to decide on the number of trees. Therefore, in the future, our main goal will be to find an efficient model selection routine to fix the number of trees. Moreover, in this article, we focus on bivariate copulas with a single conditioning variable. It will be an immediate objective to extend the approach to multivariate copulas in the presence of multiple conditioning variables to generalise the approach. Last but not the least, as we stated earlier, our sampling routine requires very few restrictions on the type of data we are using, therefore, it will be interesting to see the performance of our proposed method for more general classes of problems involving BART, where other methods may struggle. 
	
	\section{Disclosure statement}\label{disclosure-statement}
	
	The authors report that no conflicts of interest exist.
	
	\section{Data Availability Statement}\label{data-availability-statement}
	
	The CIA world factbook data was obtained from `Kaggle' repository (\url{https://www.kaggle.com/datasets/lucafrance/the-world-factbook-by-cia}).
	
	\spacingset{1.65}
	
	\bibliography{bibliography.bib}
	\appendix
    
	\section{Description of the RJ-MCMC algorithm}\label{app:rj-mcmc}
	
	We define our reversible jump MCMC algorithm based on the tree MCMC algorithm suggested by \cite{chipman98BCART} coupled with the generation of new terminal node values such that for the $k$-th tree at the $\eta$-th iteration, the proposal can be written in the following way:
	\begin{align}\label{eq:prop}
		q\left(T_k^{\eta},M_k^{\eta};T_k^{\ast}, M_k^{\ast}\right) = q\left(T_k^{\eta};T_k^{\ast}\right) q_{\left(T_k^{\eta};T_k^{\ast}\right)}\left(M_k^{\eta};M_k^{\ast}\right).
	\end{align}
	
	For each tree $T$, we associate four different sets of nodes that are crucial for tree MCMC algorithm. 
	
	\begin{itemize}
		\item $\mathcal{TN}(T_k)$: the set of all terminal nodes of a tree
		\item $\mathcal{IN}(T_k)$: the set of internal nodes
		\item $\mathcal{PN}(T_k)$: the set of all prunable nodes. That is all internal nodes which have two terminal nodes as children
		\item $\mathcal{PCN}(T_k)$: the set of all internal parent-child pairs.
	\end{itemize}
	
	To sample a new tree, we use four different tree moves \textsc{grow}, \textsc{prune}, \textsc{change} and \textsc{swap}. We define these moves in the following way:
	
	\begin{itemize}
		\item \textsc{grow}: We randomly choose a terminal node $j$ from $\mathcal{TN}(T_k)$ and grow it to two different terminal nodes $j_l$ and $j_r$ with splitting rule $x\le c$ sampled from $\pi_{\textsc{rule}}(l\mid c)$. 
		\item \textsc{prune}: We randomly choose an internal node $j$ from $\mathcal{PN}(T_k)$ and delete its children.
		\item \textsc{change}: We randomly choose an internal node $j$ from $\mathcal{IN}(T_k)$ and assign a new splitting rule.
		\item \textsc{swap}: We randomly choose a parent-child pair $(a,b)$ from $\mathcal{PCN}(T_k)$ and swap their splitting rules.
	\end{itemize}
	We illustrate these tree MCMC moves in \cref{fig:tree:MCMC}.
	
	\begin{figure}[h]
		\centering
		\begin{subfigure}{\textwidth}
			\centering
			\begin{tikzpicture}[every node/.style={font=\footnotesize}]
				\begin{scope}[shift={(-4.5cm,0)}] 
					\node (left) {
						\begin{tikzpicture}[
							level 1/.style={level distance=8mm, sibling distance=40mm},
							level 2/.style={level distance=8mm, sibling distance=20mm},
							treenode/.style={rectangle, draw=black, rounded corners,
								minimum width=14mm, minimum height=6mm, align=center},
							internal/.style={treenode, fill=blue!10},
							term/.style={treenode, fill=gray!10}]
							\node[internal] {$x_1 < c_1$}
							child { node[internal] {$x_2 < c_2$}
								child { node[term] {$\mu_1$} }
								child { node[term] {$\mu_2$} }
							}
							child { node[term] {$\mu$} };
						\end{tikzpicture}
					};
				\end{scope}
				
				\node at (0,0) {\Huge$\Longrightarrow$};
				
				\begin{scope}[shift={(4.5cm,0)}] 
					\node (right) {
						\begin{tikzpicture}[
							level 1/.style={level distance=8mm, sibling distance=40mm},
							level 2/.style={level distance=8mm, sibling distance=20mm},
							treenode/.style={rectangle, draw=black, rounded corners,
								minimum width=14mm, minimum height=6mm, align=center},
							internal/.style={treenode, fill=blue!10},
							term/.style={treenode, fill=gray!10}]
							\node[internal] {$x_1 < c_1$}
							child { node[internal] {$x_2 < c_2$}
								child { node[term] {$\mu_1$} }
								child { node[term] {$\mu_2$} }
							}
							child { node[internal] {$x_3 < c_3$}
								child { node[term] {$\mu_L$} }
								child { node[term] {$\mu_R$} }
							};
						\end{tikzpicture}
					};
				\end{scope}
			\end{tikzpicture}
			\caption{Original tree (left) and grown tree (right).}
		\end{subfigure}
		
		\begin{subfigure}{\textwidth}
			\centering
			\begin{tikzpicture}[every node/.style={font=\footnotesize}]
				\node (arrow) at (0,0) {\Huge$\Longrightarrow$};
				
				\begin{scope}[shift={(-4.5cm,0)}] 
					\node {
						\begin{tikzpicture}[
							level 1/.style={level distance=8mm, sibling distance=40mm},
							level 2/.style={level distance=8mm, sibling distance=20mm},
							treenode/.style={rectangle, draw=black, rounded corners, minimum width=14mm, minimum height=6mm, align=center},
							internal/.style={treenode, fill=blue!10},
							term/.style={treenode, fill=gray!10}]
							\node[internal] {$x_1 < c_1$}
							child { node[internal] {$x_2 < c_2$}
								child { node[term] {$\mu_1$} }
								child { node[term] {$\mu_2$} }
							}
							child { node[internal] {$x_3 < c_3$}
								child { node[term] {$\mu_L$} }
								child { node[term] {$\mu_R$} }
							};
						\end{tikzpicture}
					};
				\end{scope}
				
				\begin{scope}[shift={(4.5cm,0)}] 
					\node {
						\begin{tikzpicture}[
							level 1/.style={level distance=8mm, sibling distance=40mm},
							level 2/.style={level distance=8mm, sibling distance=20mm},
							treenode/.style={rectangle, draw=black, rounded corners, minimum width=14mm, minimum height=6mm, align=center},
							internal/.style={treenode, fill=blue!10},
							term/.style={treenode, fill=gray!10}]
							\node[internal] {$x_1 < c_1$}
							child { node[internal] {$x_2 < c_2$}
								child { node[term] {$\mu_1$} }
								child { node[term] {$\mu_2$} }
							}
							child { node[term] {$\mu$} };
						\end{tikzpicture}
					};
				\end{scope}
			\end{tikzpicture}
			\caption{Full tree (left) and pruned tree (right).}
		\end{subfigure}
		
		\begin{subfigure}{\textwidth}
			\centering
			\begin{tikzpicture}[every node/.style={font=\footnotesize}]
				\node (arrow) at (0,0) {\Huge$\Longrightarrow$};
				
				\begin{scope}[shift={(-3.5cm,0)}] 
					\node {
						\begin{tikzpicture}[
							level 1/.style={level distance=8mm, sibling distance=40mm},
							level 2/.style={level distance=8mm, sibling distance=20mm},
							treenode/.style={rectangle, draw=black, rounded corners,
								minimum width=14mm, minimum height=6mm, align=center},
							internal/.style={treenode, fill=blue!10},
							term/.style={treenode, fill=gray!10}]
							\node[internal] {$x_3 < 0.5$}
							child { node[term] {$\mu_L$} }
							child { node[term] {$\mu_R$} };
						\end{tikzpicture}
					};
				\end{scope}
				
				\begin{scope}[shift={(3.5cm,0)}] 
					\node {
						\begin{tikzpicture}[
							level 1/.style={level distance=8mm, sibling distance=40mm},
							level 2/.style={level distance=8mm, sibling distance=20mm},
							treenode/.style={rectangle, draw=black, rounded corners,
								minimum width=14mm, minimum height=6mm, align=center},
							internal/.style={treenode, fill=blue!10},
							term/.style={treenode, fill=gray!10}]
							\node[internal] {$x_3 < 0.8$}
							child { node[term] {$\mu_L$} }
							child { node[term] {$\mu_R$} };
						\end{tikzpicture}
					};
				\end{scope}
			\end{tikzpicture}
			\caption{Change of split threshold from $x_3 < 0.5$ to $x_3 < 0.8$.}
		\end{subfigure}
		
		\begin{subfigure}{\textwidth}
			
			\centering
			\begin{tikzpicture}[every node/.style={font=\footnotesize}]
				\node (arrow) at (0,0) {\Huge$\Longrightarrow$};
				
				\begin{scope}[shift={(-4.5cm,0)}] 
					\node {
						\begin{tikzpicture}[level 1/.style={level distance=8mm, sibling distance=40mm},
							level 2/.style={level distance=8mm, sibling distance=20mm},
							treenode/.style={rectangle, draw=black, rounded corners, minimum width=14mm, minimum height=6mm, align=center},
							internal/.style={treenode, fill=blue!10},
							term/.style={treenode, fill=gray!10}]
							\node[internal] {$x_1 < c_1$}
							child { node[internal] {$x_2 < c_2$}
								child { node[term] {$\mu_1$} }
								child { node[term] {$\mu_2$} }
							}
							child { node[term] {$\mu_3$} };
						\end{tikzpicture}
					};
				\end{scope}
				
				\begin{scope}[shift={(4.5cm,0)}] 
					\node {
						\begin{tikzpicture}[level 1/.style={level distance=8mm, sibling distance=40mm},
							level 2/.style={level distance=8mm, sibling distance=20mm},
							treenode/.style={rectangle, draw=black, rounded corners, minimum width=14mm, minimum height=6mm, align=center},
							internal/.style={treenode, fill=blue!10},
							term/.style={treenode, fill=gray!10}]
							\node[internal] {$x_2 < c_2$}
							child { node[internal] {$x_1 < c_1$}
								child { node[term] {$\mu_1$} }
								child { node[term] {$\mu_2$} }
							}
							child { node[term] {$\mu_3$} };
						\end{tikzpicture}
					};
				\end{scope}
			\end{tikzpicture}
			\caption{Original tree (left) and swapped version (right).}
		\end{subfigure}
		
		\caption{Graphical representation of tree MCMC moves}
		\label{fig:tree:MCMC}
	\end{figure}
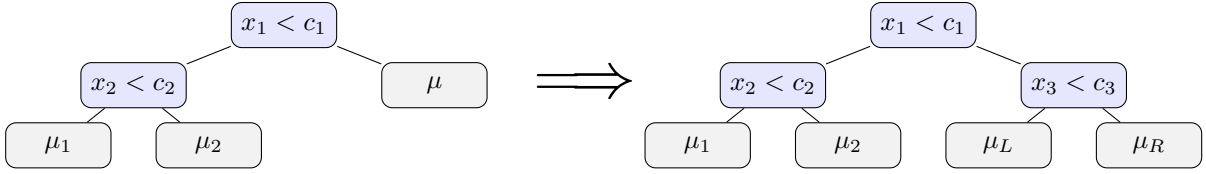
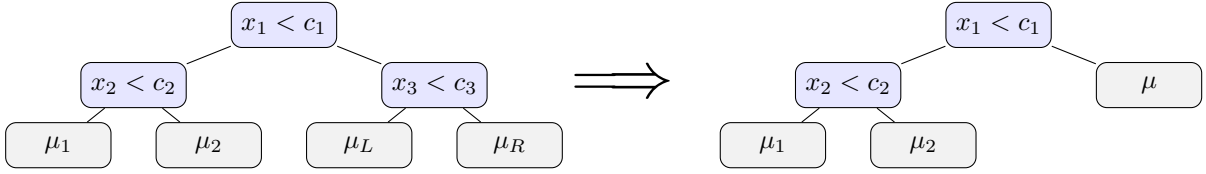
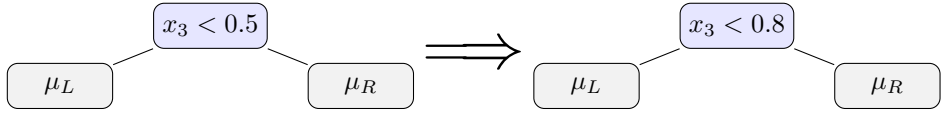
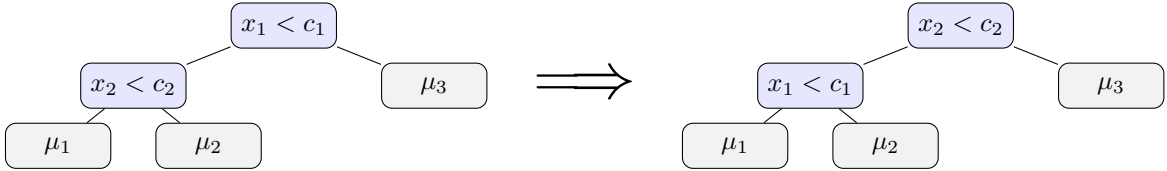

	\subsection{Proposal for each moves}
	
	Here, we explain the proposals of our RJ-MCMC algorithm for different tree moves. We follow the framework of \cite{green_RJMCMC} to define the acceptance probabilities corresponding to each move.

	\subsubsection{\textsc{grow}/\textsc{prune} move:}
	For growing a tree $T^{\ast}_k$ from $T_k$, we first randomly sample a terminal node $j$ from the set of all terminal nodes. Followed by a cut point for the splitting rule of the terminal node. Lastly, we sample terminal node values for the new left and right terminal nodes. Then we have the following:
	\begin{equation*}
		q_{k;\gamma;\textsc{grow}}\left(T_k,M_k;T_k^{\ast}, M_k^{\ast}\right) 
		= \frac{1}{\#\{\mathcal{TN}(T_k)\}} \pi_{\textsc{rule}}(j\mid c)\mathcal{N}\left(\mu_{jr}^{\ast};\mu_{j},\gamma_{jr}\right)\mathcal{N}\left(\mu_{jl}^{\ast};\mu_{j},\gamma_{jl}\right),
	\end{equation*} 
	where $\mathcal{N}(a;b,c^2)$ is the probability of getting $a$ from a normal distribution with mean $b$ and variance $c^2$. Here $\mu_j$ is same for both the proposals as the value at observations remain same after growing. However, the variance contribution changes.
	
	To go back from $T^{\ast}_k$ to $T_k$, we prune the $jl$-th terminal node and $jr$-th terminal node. For that, we calculate the probability of choosing the $j$-th node that becomes an internal node now. Followed by the proposal for terminal node value. Then,
	\begin{equation*}
		q_{k;\gamma;\textsc{prune}}\left(T_k^{\ast}, M_k^{\ast};T_k,M_k\right) 
		= \frac{1}{\#\{\mathcal{PN}(T_k^{\ast})\}} \mathcal{N}\left(\mu_j;\mu^{\ast}_j,\gamma_j\right)
	\end{equation*} 
	where $\mu^{\ast}_j = \frac{ \#\{\mathcal{I}^k_{jr}\}\mu_{jr}^{\ast} +  \#\{\mathcal{I}^k_{jl}\}\mu_{jl}^{\ast}}{ \#\{\mathcal{I}^k_{jr}\}+ \#\{\mathcal{I}^k_{jl}\}}$. Therefore, the proposal ratio of the \textsc{growth} move is given by:
	\begin{equation*}
		\frac{q_{k;\gamma;\textsc{prune}}\left(T_k^{\ast}, M_k^{\ast};T_k,M_k\right)}{q_{k;\gamma;\textsc{grow}}\left(T_k,M_k;T_k^{\ast}, M_k^{\ast}\right) } 
		= \frac{\frac{1}{\#\{\mathcal{PN}(T_k^{\ast})\}} \mathcal{N}\left(\mu_j;\mu_{j}^{\ast},\gamma_j\right)}{\frac{1}{\#\{\mathcal{TN}(T_k)\}} \pi_{\textsc{rule}}(j\mid c)\mathcal{N}\left(\mu_{jr}^{\ast};\mu_{j},\gamma_{jr}\right)\mathcal{N}\left(\mu_{jl}^{\ast};\mu_{j},\gamma_{jl}\right)}.
	\end{equation*}
	Similarly, we get the prior ratio for a \textsc{growth} move. Note that when a tree is grown at the $j$-th terminal node, we need additional splitting rule at that node. Therefore the ratio is given by:
	\begin{equation*}
		\frac{\pi(T_k^{\ast},M_k^{\ast})\pi_{\textsc{rule}}(j\mid c)}{\pi(T_k,M_k)}.
	\end{equation*}
	Finally, combining with the likelihood we get
	\begin{align*}
		&AR_{k;\gamma;\textsc{grow}}\left(T_k,M_k;T_k^{\ast}, M_k^{\ast}\right)\\
		&= \frac{
			\frac{\pi(T_k^{\ast},M_k^{\ast})}{\#\{\mathcal{PN}(T_k^{\ast})\}} 
			\prod_{i\in \mathcal{I}_{jr}^k}c\left(u_{1i},u_{2i}\mid h\left(R_{ik}+\mu^{\ast}_{jr}\right)\right)
			\prod_{i\in \mathcal{I}_{jl}^k}c\left(u_{1i},u_{2i}\mid h\left(R_{ik}+\mu^{\ast}_{jl}\right)\right)
			\mathcal{N}\left(\mu_j;\mu_{j}^{\ast},\gamma_j\right)}
		{\frac{\pi(T_k,M_k)}{\#\{\mathcal{TN}(T_k)\}} 
			\prod_{i\in \mathcal{I}_{j}^k}c\left(u_{1i},u_{2i}\mid h\left(R_{ik}+\mu_j\right)\right) 
			\mathcal{N}\left(\mu_{jr}^{\ast};\mu_{j},\gamma_{jr}\right)\mathcal{N}\left(\mu_{jl}^{\ast};\mu_{j},\gamma_{jl}\right)}.
	\end{align*}
	Since \textsc{prune} is the reverse process of the \textsc{grow} move, we have
	\begin{align*}
		&AR_{k;\gamma;\textsc{prune}}\left(T_k,M_k;T_k^{\ast}, M_k^{\ast}\right)\\
		&= \frac{
			\frac{\pi(T^{\ast}_k,M^{\ast}_k)}{\#\{\mathcal{TN}(T^{\ast}_k)\}} 
			\prod_{i\in \mathcal{I}_{j}^k}c\left(u_{1i},u_{2i}\mid h\left(R_{ik}+\mu^{\ast}_j\right)\right) 
			\mathcal{N}\left(\mu_{jr};\mu^{\ast}_{j},\gamma_{jr}\right)
			\mathcal{N}\left(\mu_{jl};\mu^{\ast}_{j},\gamma_{jl}\right)}
		{\frac{\pi(T_k,M_k)}{\#\{\mathcal{PN}(T_k)\}} 
			\prod_{i\in \mathcal{I}_{jr}^k}c\left(u_{1i},u_{2i}\mid h\left(R_{ik}+\mu_{jr}\right)\right)
			\prod_{i\in \mathcal{I}_{jl}^k}c\left(u_{1i},u_{2i}\mid h\left(R_{ik}+\mu_{jl}\right)\right)
			\mathcal{N}\left(\mu^{\ast}_j;\mu_{j},\gamma_j\right)}
	\end{align*}
	where $\mu_j = \frac{\#\{\mathcal{I}^k_{jr}\}\mu_{jr} + \#\{\mathcal{I}^k_{jl}\}\mu_{jl}}{\#\{\mathcal{I}^k_{jr}\} + \#\{\mathcal{I}^k_{jl}\}}$.
	
	\subsubsection{\textsc{change}/\textsc{swap} move:}
	For the \textsc{change} or the \text{swap} move. The proposals only rely on the internal nodes and the terminal nodes and values remain unchanged. Moreover, for each of these moves $\mathcal{IN}(T_k)$ and $\mathcal{PCN}(T_k)$ remain invariant in nature. That is for a proposed tree $T_k^{\ast}$, $\mathcal{IN}(T_k)=\mathcal{IN}(T_k^{\ast})$ and $\mathcal{PCN}(T_k)=\mathcal{PCN}(T_k^{\ast})$. Therefore, the proposal ratios of these two moves are equal to 1. This gives the following:
	
	\begin{align*}
		AR_{k;\gamma;\textsc{move}}\left(T_k,M_k;T_k^{\ast}, M_k^{\ast}\right)= \frac{
			\pi(T^{\ast}_k,M^{\ast}_k) 
			\prod_{i=1}^{n}c\left(u_{1i},u_{2i}\mid h\left(R_{ik}+g(x_i, T_k^{\ast}, M_k^{\ast})\right)\right)}
		{\pi(T_k,M_k) 
			\prod_{i=1}^{n}c\left(u_{1i},u_{2i}\mid h\left(R_{ik}+g(x_i, T_k, M_k)\right)\right)}
	\end{align*}
	for $\textsc{move}\in\{\textsc{change},\textsc{swap}\}$.
	
	Once we have the acceptance ratio, we can calculate the acceptance probability $\alpha_{k;\gamma;\textsc{move}}\left(T_k,M_k;T_k^{\ast}, M_k^{\ast}\right)$ of moving from $(T_k,M_k)$ to $(T_k^{\ast},M_k^{\ast})$ in the following way:
	\begin{equation*}
		\alpha_{k;\gamma;\textsc{move}}\left(T_k,M_k;T_k^{\ast}, M_k^{\ast}\right) = \min\left\{1,AR_{k;\gamma;\textsc{move}}\left(T_k,M_k;T_k^{\ast}, M_k^{\ast}\right)\right\},
	\end{equation*}
	where $\textsc{move}\in \left\{\textsc{grow}, \textsc{prune}, \textsc{change}, \textsc{swap}\right\}$.

	\subsection{Markov transition kernel of the RJ-MCMC step}
	
	Let $\mathcal{S}\coloneqq\mathcal{T}\times \mathcal{M}$ be the state space and $S_k\coloneqq(T_k,M_k) \in \mathcal{S}$ be the current state of the $k$-th tree. Let $p_{\textsc{g}}$, $p_{\textsc{p}}$, $p_{\textsc{c}}$ and $p_{\textsc{s}}$ be the probabilities of selecting \textsc{grow}, \textsc{prune}, \textsc{change} and \textsc{swap} move respectively. Then dropping the tree index $k$ for convenience, we can write the transition kernel of the RJ-MCMC algorithm with $\gamma$ adaption in the following way:
	\begin{align*}
		\mathcal{K}_{\gamma}\left(s; \cdot\right) 
		= p_{\textsc{g}}\mathcal{K}_{\gamma;\textsc{grow}}\left(s; \cdot\right) + 
		p_{\textsc{p}}\mathcal{K}_{\gamma;\textsc{prune}}\left(s; \cdot\right) + 
		p_{\textsc{c}}\mathcal{K}_{\gamma;\textsc{change}}\left(s; \cdot\right)+ 
		p_{\textsc{s}}\mathcal{K}_{\gamma;\textsc{swap}}\left(s; \cdot\right).
	\end{align*}
	Here, for the \textsc{change} and \textsc{swap} moves, we do not need to propose any new terminal node values hence $\gamma$ adaption is not utilised and the notation is used for the sake of consistency. Now, let $s^{\ast}$ be a proposed state, then each of these kernels with $\gamma$ adaption are given by
	\begin{align*}
		\mathcal{K}_{\gamma;\textsc{move}}\left(s, \cdot\right) 
		&= \alpha_{\gamma;\textsc{move}}\left(s;s^{\ast}\right) 
		q_{\gamma;\textsc{move}}\left(s;s^{\ast}\right)ds^{\ast} \\
		&\qquad+ \left(1 - \int_{\mathcal{T}\times\mathcal{M}} \alpha_{\gamma;\textsc{move}}\left(s;s^{\ast}\right) 
		q_{\gamma;\textsc{move}}\left(s;s^{\ast}\right)ds^{\ast}\right) \delta_{s}(ds^{\ast}),
	\end{align*}
	where $\textsc{move}\in \left\{\textsc{grow}, \textsc{prune}, \textsc{change}, \textsc{swap}\right\}$ and $\delta_x(dy)$ is the Dirac measure at $x$. 
	
	\subsubsection{Adaptive proposal variance} 
	
	For the sake of notational consistency, we restate the algorithm for computing the adaptive proposal variance in \cref{alg:ada:prop} without the tree index.
	
	\begin{algorithm}[h]
		\caption{Computation of adaptive proposal}\label{alg:ada:prop}
		\begin{algorithmic}[1]
			\State Perform $\eta_0$ iterations with a fixed variance.
			
			\State Collect MCMC samples for initial $\eta_0$ iterations:
			\begin{equation*}
				V_{i}^{\eta} = g\left(x_i,T^{\eta},M^{\eta}\right)\quad 1\le i \le n; 1\le \eta \le \eta_0.
			\end{equation*}
			
			\State Calculate sample covariance matrix :
			\begin{equation*}
				C^{\eta} \coloneqq Cov\left(V^{1},V^{2},\cdots,V^{\eta}\right) + \epsilon\mathbf{I}_n \quad \text{for } \epsilon>0; \eta>\eta_0.
			\end{equation*}
			
			\State Calculate the proposal variance of the $j$-th terminal node value:
			\begin{equation*}
				\sigma_{\text{prop};j}^2 \coloneqq \frac{2.4^2}{\left(\#\left\{\mathcal{I}_{j}\right\}\right)^3}
				\sum_{c\in\mathcal{I}_{j}}\sum_{d\in\mathcal{I}_{j}}\left[C^{\eta}\right]_{cd};\qquad\mathcal{I}_{j}\coloneqq \left\{i:x_i\in \Omega_{j}\right\};\qquad 1\le j\le n_L(T).
			\end{equation*}
			
			\State Update sample covariance using 
			\begin{equation*}
				C^{\eta+1} \coloneqq 
				\frac{\eta-1}{\eta} C^{\eta} 
				+ \frac{1}{\eta}\left(\eta 
				\left(\overline{V}^{\eta}\right)\left(\overline{V}^{\eta}\right)^T 
				- (\eta+1)\left(\overline{V}^{\eta+1}\right)\left(\overline{V}^{\eta+1}\right)^T 
				+ \left({V}^{\eta+1}\right)\left({V}^{\eta+1}\right)^T + \epsilon\mathbf{I}_n\right).
			\end{equation*}
		\end{algorithmic}
	\end{algorithm}

	\subsection{Proof of Ergodicity}
	
	Before providing the proof of ergodicity. We restate the required assumptions and prove two lemmas that are necessary for showing ergodicity.
	
	\begin{assumption}\label{ass:bounded:cop}
		$\exists$ \(\delta>0\) such that the copula density is bounded. That is, $0 < c(u_1,u_2\mid\theta(h(g(x,T,M)))) < \infty$ for all $(u_1,u_2)\in[\delta,1-\delta]^2$ and for all $M \in \mathcal{M}$
	\end{assumption}
	\begin{assumption}\label{ass:bounded:prior}
		All prior densities for the terminal node values are bounded above.
	\end{assumption}
	\begin{assumption}\label{ass:bounded:prop}
		All proposal densities of the terminal node values are uniformly bounded.
	\end{assumption}
	\begin{assumption}\label{ass:bounded:tree}
		With $n$ observations, the number of terminal nodes $n_L(T)\le n$ and for any tree with $n_L(T)\ge 2$, there exists at least one internal node that can be pruned.
	\end{assumption}

	Now, let $S\coloneqq \left(T,M\right)$ be a state in the state space $\mathcal{S}\coloneqq\mathcal{T}\times \mathcal{M}$. Let $\Gamma^{\eta} \coloneqq \left(\Gamma_{1}^{\eta}, \cdots\Gamma^{\eta}_{n_L(T)}\right)\in\mathcal{Y}$ be the adaption at the $\eta$-th iteration such that for $j \in \{1,2,\cdots, n_L(T)\}$
	\begin{equation*}
		\Gamma^{\eta}_{j} \coloneqq \frac{2.4^2}{\left(\#\left\{\mathcal{I}_{j}\right\}\right)^3}
		\sum_{c\in\mathcal{I}_{j}}\sum_{d\in\mathcal{I}_{j}}\left[C^{\eta}\right]_{cd}.
	\end{equation*}
	Then the following two lemmas are true.
	
	\begin{lemma}\label{lemm:bounded:ratio}
		If \cref{ass:bounded:cop}-\cref{ass:bounded:prop} are satisfied then the acceptance ratios are uniformly bounded. That is, there exist values $0< b_{\textsc{move}} <  B_{\textsc{move}} < \infty$ such that
		\begin{equation*}
			\frac{1}{B_{\textsc{move}}} \le AR_{\gamma, \textsc{move}}\left(s;s^{\ast}\right) \le \frac{1}{b_{\textsc{move}}}
		\end{equation*}
		$\forall\gamma \in \mathcal{Y}$.
	\end{lemma}
	
	\begin{proof}
		In a compact support of $[\delta,1-\delta]^2$, the bivariate copula density functions presented in this article are bounded. Besides this, the prior and proposal for a tree is from a discrete space and the priors and proposals for the terminal node values have finite variances. Therefore, the posterior is bounded for all possible terminal node values in $\mathcal{M}$. Moreover, since the proposal variances are uniformly bounded \citep{haario_AMH}, the bounds of the acceptance ratios are independent of the $\gamma$ adaption.
	\end{proof}
	
	This lemma ensures that we do not have degenerate values for the acceptance rate which is otherwise problematic for the convergence of MCMC algorithm.
	
	\begin{lemma}\label{lemm:stable:var}
		$\left\|\Gamma^{\eta+1}_{j} - \Gamma^{\eta}_{j}\right\|_{\infty}\to 0$ as $\eta\to\infty$.
	\end{lemma}
	
	\begin{proof}
		Recall the iterative variance formula given by:
		\begin{equation*}
			C^{\eta+1} \coloneqq 
			\frac{\eta-1}{\eta} C^{\eta} 
			+ \frac{1}{\eta}\left(\eta 
			\left(\overline{V}^{\eta}\right)\left(\overline{V}^{\eta}\right)^T 
			- (\eta+1)\left(\overline{V}^{\eta+1}\right)\left(\overline{V}^{\eta+1}\right)^T 
			+ \left({V}^{\eta+1}\right)\left({V}^{\eta+1}\right)^T + \epsilon\mathbf{I}_n\right),
		\end{equation*}
		where $\overline{V}^{\eta}$ is the point-wise mean of $\eta$ MCMC samples. That is, for $n$ number of observations $\overline{V}^{\eta} \coloneqq \left(\frac{1}{\eta}\sum_{i=1}^\eta V_1^i,\cdots, \frac{1}{\eta}\sum_{i=1}^\eta V_n^i\right)^T$. Using this iterative formula, we define the change in variance in two consecutive iterations in the following way:
		\begin{equation*}
			\Delta C^{\eta + 1} \coloneqq \frac{1}{\eta}\left(\eta \left(\overline{V}_{\cdot }^{\eta}\right)\left(\overline{V}_{\cdot }^{\eta}\right)^T 
			- (\eta+1)\left(\overline{V}_{\cdot }^{\eta+1}\right)\left(\overline{V}_{\cdot }^{\eta+1}\right)^T 
			+ \left({V}_{\cdot }^{\eta+1}\right)\left({V}_{\cdot }^{\eta+1}\right)^T + \epsilon\mathbf{I}_n - C^{\eta}\right).
		\end{equation*}
		Now, let $ n_j = \#\left\{\mathcal{I}_{j}\right\}$ denote the number of observations contained in the $j$-th partition. Then
		\begin{align*}
			\Gamma^{\eta+1}_j 
			= \frac{2.4^2}{\left(n_{j}\right)^3}
			\sum_{c\in\mathcal{I}_{j}}\sum_{d\in\mathcal{I}_{j}}\left[C^{\eta+1}\right]_{cd}
			&= \frac{2.4^2}{\left(n_{j}\right)^3}
			\left[\sum_{c\in\mathcal{I}_{j}}\sum_{d\in\mathcal{I}_{j}}\left[C^{\eta} + \Delta C^{\eta + 1}\right]_{cd}\right] \\
			&= \Gamma^{\eta}_j + \frac{2.4^2}{\left(n_{j}\right)^3} \left[\sum_{c\in\mathcal{I}_{j}}\sum_{d\in\mathcal{I}_{j}}\left[\Delta C^{\eta + 1}\right]_{cd}\right].
		\end{align*}
		Since $\overline{V}^{\eta} \coloneqq \left(\frac{1}{\eta}\sum_{i=1}^\eta V_1^i,\cdots, \frac{1}{\eta}\sum_{i=1}^\eta V_n^i\right)^T$ and $\overline{V}^{\eta+1} \coloneqq \left(\frac{1}{\eta+1}\sum_{i=1}^{\eta+1} V_1^i,\cdots, \frac{1}{\eta+1}\sum_{i=1}^{\eta+1} V_n^i\right)^T$, the increment is of $O(1/\eta)$ in the $(\eta +1)$-th iteration. Therefore $\left\|\Gamma^{\eta+1}_j - \Gamma^{\eta}_j\right\|_{\infty} \to 0$ as $\eta\to \infty$.
	\end{proof}
	This way, we ensure that variance in each iterative steps does not change rapidly and eventually the change becomes infinitesimally small.
	
	\begin{manualtheorem}{4.1}
		If \cref{ass:bounded:cop}-\cref{ass:bounded:tree} are satisfied and $\gamma$ adaption strategy is defined by \cref{alg:ada:prop}. Then the following two conditions hold:
		\begin{enumerate}
			\item Diminishing adaptation $\rightarrow$ $\left\|\mathcal{K}_{\Gamma^{\eta+1}}(s;\cdot)-\mathcal{K}_{\Gamma^{\eta}}(s;\cdot)\right\|_{TV}\to 0$ in probability as $\eta\to \infty$.
			\item Containment $\rightarrow$ $\forall\iota>0$, $\exists N = N(\iota)\in \mathbb{N}$ such that 
			\begin{equation*}
				\left\|\mathcal{K}^N_{\gamma}(s;\cdot) - \pi(\cdot\mid u_1,u_2)\right\|_{TV} \le \iota
			\end{equation*}
		\end{enumerate}
		where $\pi(\cdot\mid u_1, u_2)$ denote the target distribution defined on $\mathcal{S}$.
	\end{manualtheorem}
	
	\begin{proof}
		\textbf{Diminishing adaptation:}
		To show diminishing adaptation, we check the total variation norm for each move. Let the \textsc{grow} move is denoted by \textsc{g}. Then the kernel corresponding \textsc{grow} gives us
		\begin{align*}
			&\left\|\mathcal{K}_{\Gamma^{\eta+1};\textsc{g}}(s;\cdot)-\mathcal{K}_{\Gamma^{\eta};\textsc{g}}(s;\cdot)\right\|_{TV}\\
			&\le 2\int \left|\alpha_{\Gamma^{\eta+1};\textsc{g}}(s;s^{\ast})q_{\Gamma^{\eta+1};\textsc{g}}(s;s^{\ast}) - \alpha_{\Gamma^{\eta};\textsc{g}}(s;s^{\ast})q_{\Gamma^{\eta};\textsc{g}}(s;s^{\ast})\right|ds^{\ast}\\
			&\le 2\int \left|\alpha_{\Gamma^{\eta+1};\textsc{g}}(s;s^{\ast})\left(q_{\Gamma^{\eta+1};\textsc{g}}(s;s^{\ast}) - q_{\Gamma^{\eta};\textsc{g}}(s;s^{\ast})\right)\right|ds^{\ast} 
			+ 2\int \left|\left(\alpha_{\Gamma^{\eta+1};\textsc{g}}(s;s^{\ast})-\alpha_{\Gamma^{\eta};\textsc{g}}(s;s^{\ast})\right)q_{\Gamma^{\eta};\textsc{g}}(s;s^{\ast})\right|ds^{\ast}\\
			&\le 2\int \left|q_{\Gamma^{\eta+1};\textsc{g}}(s;s^{\ast}) - q_{\Gamma^{\eta};\textsc{g}}(s;s^{\ast})\right|ds^{\ast} 
			+ 2\int \left|\left(\alpha_{\Gamma^{\eta+1};\textsc{g}}(s;s^{\ast})-\alpha_{\Gamma^{\eta};\textsc{g}}(s;s^{\ast})\right)q_{\Gamma^{\eta};\textsc{g}}(s;s^{\ast})\right|ds^{\ast}.
		\end{align*}
		
		Since $q_{\Gamma^{\eta+1};\textsc{g}}(s;s^{\ast})$ is a product of Gaussian probabilities with bounded variances and the adaption $\Gamma^{\eta+1}$ converges to $\Gamma^{\eta}$ by \cref{lemm:stable:var}, we can show that $\left\|q_{\Gamma^{\eta+1};\textsc{g}}(s;s^{\ast})-q_{\Gamma^{\eta};\textsc{g}}(s;s^{\ast})\right\|_{TV}$ goes to 0 as $\eta \to \infty$. Similarly, from \cref{lemm:bounded:ratio}, we know that the acceptance ratio and the proposal are uniformly bounded and the sequence $\alpha_{\gamma;\textsc{g}}(s;s^{\ast})$ is trivially dominated by 1. Hence exploiting the pointwise convergence of $\Gamma^{\eta}$ and continuity of $\alpha_{\Gamma^{\eta+1}}(s;s^{\ast})$, we can apply dominated convergence theorem and show that $\left\|\left(\alpha_{\Gamma^{\eta+1};\textsc{g}}(s;s^{\ast})-\alpha_{\Gamma^{\eta};\textsc{g}}(s;s^{\ast})\right)q_{\Gamma^{\eta};\textsc{g}}(s;s^{\ast})\right\|_{TV}$ goes to 0 as $\eta \to \infty$. Hence $\left\|\mathcal{K}_{\Gamma^{\eta+1};\textsc{g}}(s;\cdot)-\mathcal{K}_{\Gamma^{\eta};\textsc{g}}(s;\cdot)\right\|_{TV}$ goes to 0 as $\eta\to \infty$.
		
		With a similar argument for the \textsc{prune} move, we have $\left\|\mathcal{K}_{\Gamma^{\eta+1};\textsc{p}}(s;\cdot)-\mathcal{K}_{\Gamma^{\eta};\textsc{p}}(s;\cdot)\right\|_{TV}$ goes to 0 as $\eta\to \infty$. For \textsc{change} and \textsc{swap} moves, no adaption takes place. Therefore, we can show that the diminishing adaption holds for our proposed adaptive RJ-MCMC algorithm.
		
		\textbf{Containment:}
		To show the containment condition, we follow a similar approach to that of \citet{andrieu_RJ_MCMC}. Recall the state space denoted by $\mathcal{S}\coloneqq\mathcal{T}\times\mathcal{M}$. Then by \cref{ass:bounded:tree}, we can express this space as a finite union of subspaces such that $\mathcal{S} \coloneq \bigcup_{i=1}^n \mathcal{S}_i$, where $\mathcal{S}_i = \mathcal{T}_i \times \mathcal{M}_i$ is the space of all trees with $i$ terminal nodes and the corresponding vector of terminal nodes of length $i$ and $\mathcal{M}_i\subseteq\mathbb{R}^i$. 
		
		We will show that the Markov transition kernel $\mathcal{K}_{\gamma}(\cdot;\cdot)$ satisfies the minorisation condition \citep{meyn_tweedie1994} on the state space $\mathcal{S}$. That is the state space is a small set. First, we show that from any tree with $i$ terminal nodes, we can explore other trees with non zero probability.
		
		Let $s_i \coloneqq\left(T^{(i)}, M^{(i)}\right)$ denote a pair of tree with $i$ terminal nodes and corresponding vector of terminal node values. Then for each $i=2,\cdots,n$, and $B_{i-1}\subseteq \mathcal{S}_{i-1}$
		\begin{align*}
			\mathcal{K}_{\gamma}\left(s_i;B_{i-1}\right)
			&\ge p_{\textsc{p}}\int_{B_{i-1}}\mathcal{K}_{\gamma,\textsc{prune}}\left(s_i;ds_{i-1}\right)\\
			&\ge p_{\textsc{p}}\int_{B_{i-1}}\alpha_{\gamma;\textsc{prune}}\left(s_i;s_{i-1}\right)q_{\gamma;\textsc{prune}}\left(s_i;s_{i-1}\right)ds_{i-1}\\
			&\ge \frac{p_{\textsc{p}}}{B_{\textsc{prune}}}\int_{B_{i-1}}q_{\gamma;\textsc{prune}}\left(s_i;s_{i-1}\right)ds_{i-1}.
		\end{align*}
		Therefore there exist a uniform lower bound $\nu>0$ independent of $\gamma$ such that for all $i=2,\cdots,n$ and $\gamma\in \mathcal{Y}$ the following holds
		\begin{equation}\label{eq:reach:trivial}
			\mathcal{K}_{\gamma}\left(s_i;B_{i-1}\right)
			\ge \nu\int_{B_{i-1}}q_{\gamma;\textsc{prune}}\left(s_i;s_{i-1}\right)ds_{i-1}.
		\end{equation}
		
		Now, let $s_1 = \left(T^{(1)},M^{(1)}\right)$ denote the trivial tree. That is, a tree with only one node. Then for all $i=2,\cdots, n$, we can reach the trivial tree with non-zero probability. Hence, we can define a measure $\phi$ such that
		\begin{equation*}
			\phi(A) = \int_{\mathcal{M}_1}\mathbb{I}\left(s_1\in A\right)\mathcal{N}\left(\mu;0,\sigma_1^2\right)d\mu
		\end{equation*}
		where $\sigma^2_1$ is the prior variance for terminal node value of the trivial tree. We will use this measure to satisfy the minorisation condition. Moreover, we can use this measure to show that $\mathcal{K}_{\gamma}(\cdot;\cdot)$ is $\phi$-irreducible. 
		
		Aperiodicity of the $\mathcal{K}_{\gamma}(\cdot;\cdot)$ is also achieved through construction. For instance, we can propose $s_2$ for which $\alpha_{\gamma;\textsc{grow}}\left(s_1; s_1\right) < 1$, independent of $\gamma$. Then the rejection probability $r_{\gamma;\textsc{grow}}\left(s_1\right)$ is given by:
		\begin{align*}
			r_{\gamma;\textsc{grow}}\left(s_1\right)
			&=1-\int \alpha_{\gamma;\textsc{grow}}\left(s_1; s_2\right)q_{\gamma;\textsc{grow}}\left(s_1; s_2\right) ds_2\\
			&=\int\left(1 - \alpha_{\gamma;\textsc{grow}}\left(s_1; s_2\right)\right) q_{\gamma;\textsc{grow}}\left(s_1; s_2\right) ds_2 > 0.
		\end{align*}
		Therefore, $\forall\gamma\in\mathcal{Y}$, there exists $1>\nu_0 >0$ such that
		\begin{equation}\label{eq:stay:trivial}
			\mathcal{K}_{\gamma}\left(s_1;\left\{s_1\right\}\right) \ge \nu_0 = \nu_0 \phi\left(\left\{s_1\right\}\right).
		\end{equation}
		
		Now, to show that minorisation holds on $\mathcal{S}$, we iterate the transition kernel $n$ times and using Chapman Kolmogorov equation we get
		\begin{align*}
			\mathcal{K}_{\gamma}^n\left(s_k;s_1\right)
			= \int_{\mathcal{S}}\mathcal{K}_{\gamma}^{k-1}\left(s_k;s\right)\mathcal{K}_{\gamma}^{n+1-k}\left(s;s_1\right)ds
			& \ge \int_{\mathcal{S}_1}\mathcal{K}_{\gamma}^{k-1}\left(s_k;s\right)\mathcal{K}_{\gamma}^{n+1-k}\left(s;s_1\right)ds\\
			& = \mathcal{K}_{\gamma}^{k-1}\left(s_k;s_1\right)\mathcal{K}_{\gamma}^{n+1-k}\left(s_1;\left\{s_1\right\}\right).
		\end{align*}
		Therefore, from \cref{eq:reach:trivial} and \cref{eq:stay:trivial}, we get 
		\begin{equation*}
			\mathcal{K}_{\gamma}^n\left(s_k;s_1\right) \ge \nu^{k-1} \nu_0^{n+1-k}\phi\left(s_1\right).
		\end{equation*}
		This holds true for $k=1$ as well. Therefore, for any $s_{k'}$ in $\mathcal{S}$,  
		\begin{equation*}
			\mathcal{K}_{\gamma}^n\left(s_k;s_{k'}\right) \ge \vartheta\phi\left(s_{k'}\right)
		\end{equation*}
		where $\vartheta = \min_{1\le k\le n}\left\{\nu^{k-1} \nu_0^{n+1-k}\right\}$. This shows that the state space $\mathcal{S}$ is small \citep{meyn_tweedie1994}. Then following the proposition given by \citet{Tierney1994}, we can find a constant $L_0>0$ such that for all $N$
		\begin{equation}
			\left\|\mathcal{K}^N_{\gamma}(s;\cdot) - \pi(\cdot\mid y)\right\|_{TV} \le L_0 (1-\vartheta)^{\lfloor N/n\rfloor}.
		\end{equation}
		Therefore, $K_{\gamma}(s;)$ is uniformly ergodic for all values of $\gamma$ and $s$. Hence, for all $\iota>0$, we can find $N = n\left(\frac{\log(\iota)-\log(L_0)}{\log(1-\vartheta)} +1\right)$ such that
		\begin{equation}
			\left\|\mathcal{K}^N_{\gamma}(s;\cdot) - \pi(\cdot\mid y)\right\|_{TV} \le \iota \qquad\forall s\in \mathcal{S}, \gamma\in\mathcal{Y}.
		\end{equation}
		
	\end{proof}
	
	\section{Synthetic dataset}\label{app:results}
	
	\subsection{Tree structure datasets}
	We present the traceplots of the number of terminal nodes in \cref{fig:trace:nterm:ex1} and the depth in \cref{fig:trace:depth:ex1} for the first replicate corresponding to each copula. We notice that the tree exploration stabilises after about 1000 iterations with our proposed approaches. We also present the likelihood plot after discarding 1500 burn in samples in \cref{fig:trace:like:ex1}. We see that our method is able to converge towards the true likelihood region for all cases, shown by the black dashed line. 
	
	\begin{figure}
		\centering
		\includegraphics[width = 0.8\linewidth]{"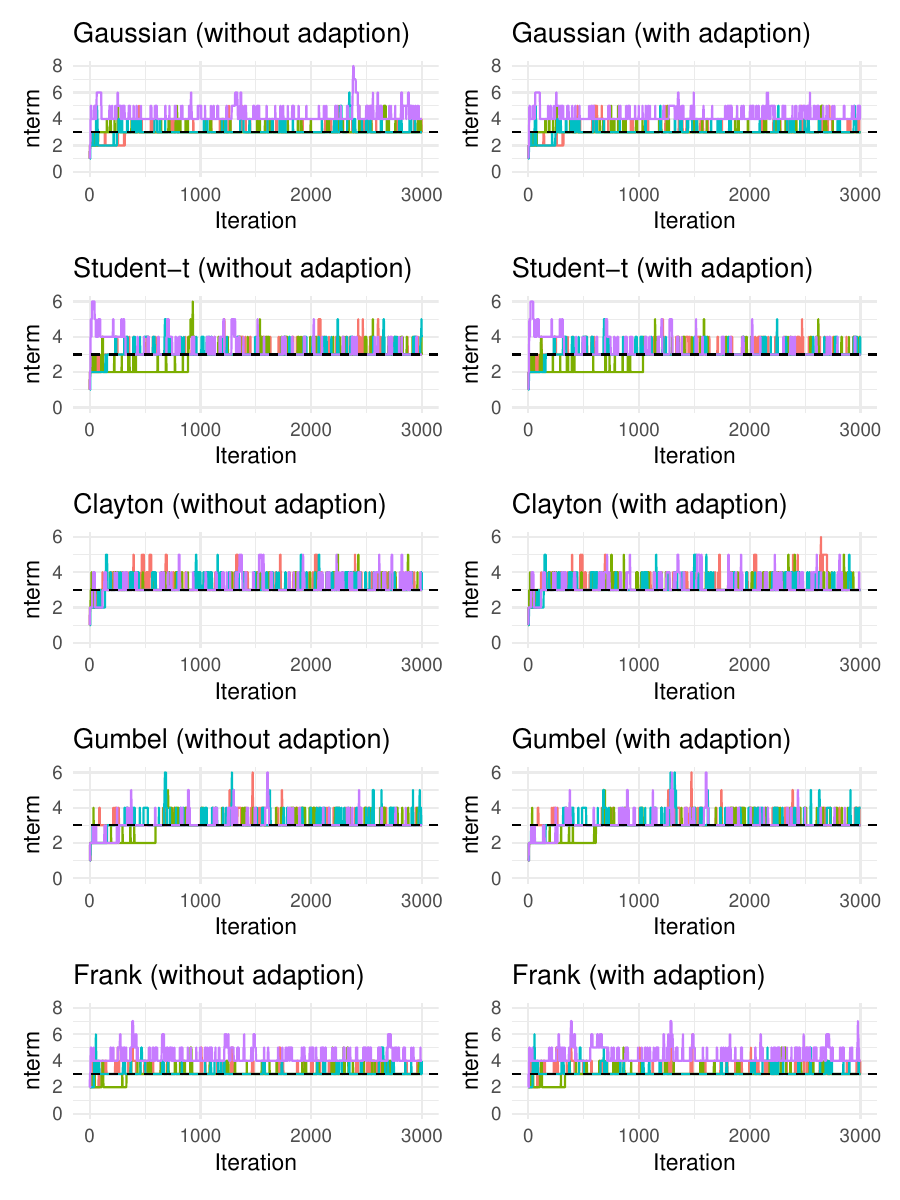"}
		\caption{Trace plots of the number of terminal nodes obtained from our analyses with synthetic datasets generated using tree based Kendall's tau, for the first replicate. The plots are obtained by running 4 parallel chains, each one with 3000 MCMC iterations. The left column shows analyses with C-BART and the right column shows analyses with A-C-BART. The horizontal dashed black line represents the number of terminal nodes (3) of the data-generating tree.}
		\label{fig:trace:nterm:ex1}
	\end{figure}
	
	\begin{figure}
		\centering
		\includegraphics[width = 0.8\linewidth]{"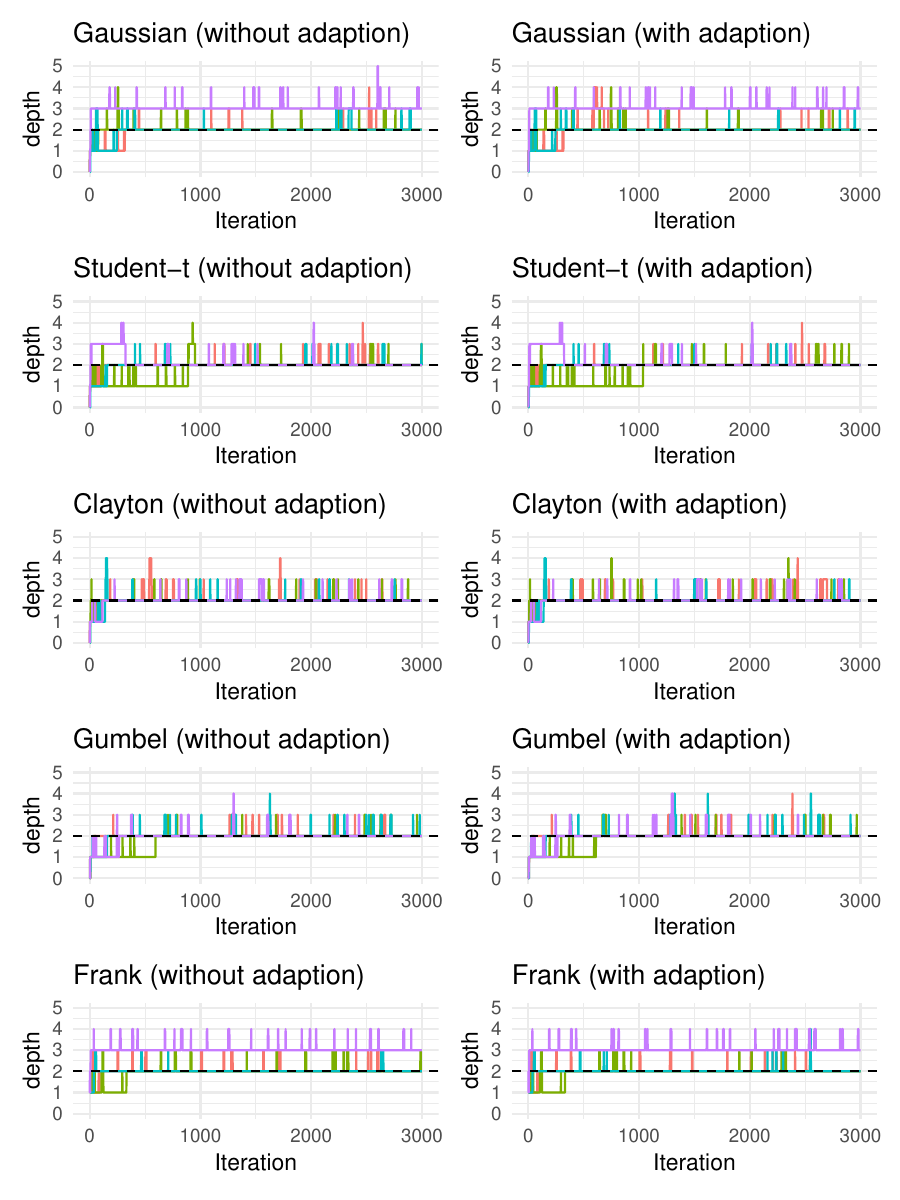"}
		\caption{Trace plots of the depth of the tree obtained from our analyses with synthetic datasets generated using tree based Kendall's tau, for the first replicate. The plots are obtained by running 4 parallel chains, each one with 3000 MCMC iterations. The left column shows analyses with C-BART and the right column shows analyses with A-C-BART. The horizontal dashed black line represents the depth (2) of the data-generating tree.}
		\label{fig:trace:depth:ex1}
	\end{figure}
	
	\begin{figure}
		\centering
		\includegraphics[width = 0.8\linewidth]{"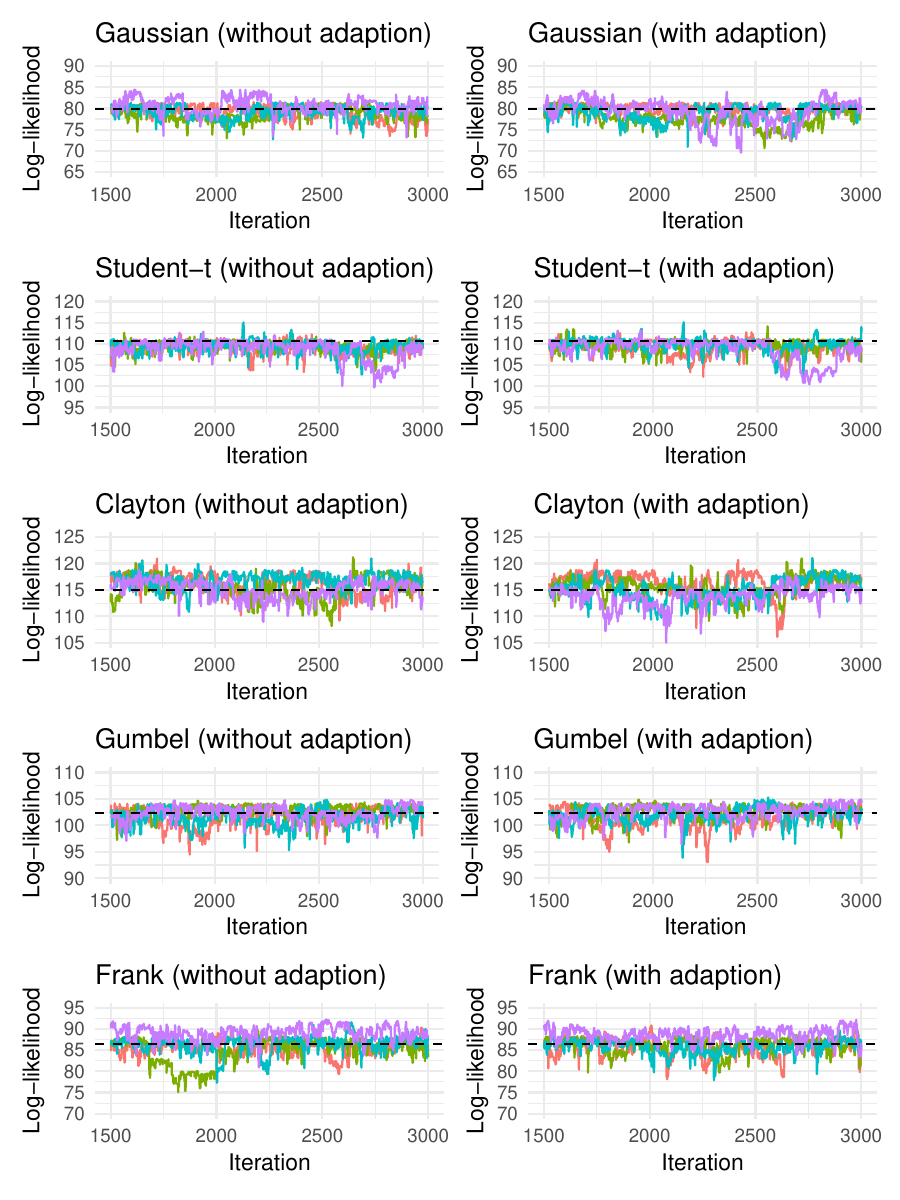"}
		\caption{Trace plots of the log-likelihood obtained from our analyses with synthetic datasets generated using tree based Kendall's tau, for the first replicate. The plots are obtained by running 4 parallel chains, each one with 3000 MCMC iterations. The left column shows analyses with C-BART and the right column shows analyses with A-C-BART. The horizontal dashed black line represents the true log-likelihood of the data-generating tree.}
		\label{fig:trace:like:ex1}
	\end{figure}
	
	We present the traceplots of the number of terminal nodes, depth and likelihood for the Frank copula in \cref{fig:trace:frank} when the proposal variance is 0.2. We notice that for the first dataset, only one out of 4 chains converge to true likelihood region for C-BART after 7500 iterations, whereas all 4 chains converge to the true likelihood region for A-C-BART. We also present the efficiency of our methods in finding the true tree structure in \cref{tab:eff:ex1:frank} and prediction accuracy in \cref{tab:pred:ex1:frank}.
	
	\begin{figure}[H]
		\centering
		\includegraphics[width = 0.85\linewidth]{"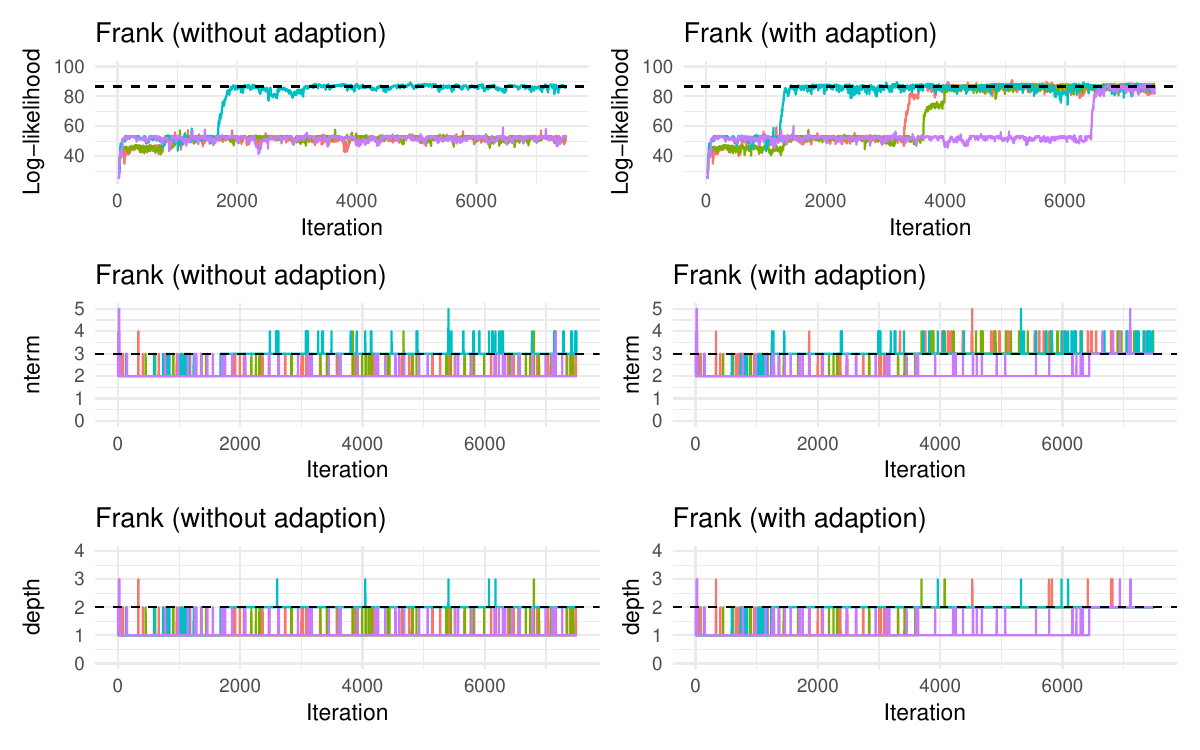"}
		\caption{Trace plots of the likelihood; the number of terminal nodes; and the depth of the tree for the first replicate of the Frank copula datasets generated using tree based Kendall's tau. The plots are obtained by running 4 parallel chains with 7500 MCMC iterations each and proposal variance fixed at $0.2$. The left column shows analyses with C-BART and the right column shows analyses with A-C-BART.}
		\label{fig:trace:frank}
	\end{figure}
	
	\begin{table}[H]
		\centering
		\begin{tabular}{l|cccccc}
			& Mean ($\hat{n}_L$) & SD ($\hat{n}_L$) & Mean ($\hat{D}$) & SD ($\hat{D}$) & Mean Acc. & SD Acc. \\ 
			\midrule
			Frank (C-BART) & 2.242 & 0.277 & 1.232 & 0.265 & 0.111 & 0.029 \\
			Frank (A-C-BART) & 2.592 & 0.343 & 1.569 & 0.329 & 0.117 & 0.024 \\
		\end{tabular}
		\caption{Performance of our proposed method for Frank copula dataset using a tree based conditional Kendall's tau. We present the average and the standard deviation of the posterior expectation of the number of terminal nodes and the depth of the tree along with the acceptance rate. The quantities are obtained by running 4 parallel chains with 7500 MCMC iterations each and proposal variance fixed at $0.2$.}
		\label{tab:eff:ex1:frank}
	\end{table}
	
	\begin{table}[H]
		\centering
		\begin{tabular}{l|ccc|ccc}
			\multicolumn{1}{c|}{} &
			\multicolumn{3}{c|}{C-BART} &
			\multicolumn{3}{c}{A-C-BART} \\
			\midrule
			& RMSE & CI-length & CI-cov & RMSE & CI-length & CI-cov \\ 
			\midrule
			Frank & 0.194 & 0.224 & 0.432 & 0.153 & 0.244 & 0.688 
		\end{tabular}
		\caption{Prediction accuracy of our proposed method for the first replicate of the Frank copula datasets using a tree based conditional Kendall's tau. We split our results in two general columns: left for C-BART, right for A-C-BART. We create subcolumns under each column to present root mean squared error (RMSE); average 95\% credible interval length (CI-length); and 95\% credible interval coverage (CI-cov). The quantities are obtained by running 4 parallel chains with 7500 MCMC iterations each and proposal variance fixed at $0.2$.}
		\label{tab:pred:ex1:frank}
	\end{table}
	
	\subsection{Dataset with a general function}
	
	We present the prediction plots corresponding to the first replicate of the copula datasets in \cref{fig:trace:pred:ex2}. In the figure, we represent the true values of Kendall's conditional tau with red points; posterior estimates of Kendall's tau with a black line and credible intervals with green lines. 
	
	Additionally, we present the likelihood traceplots in \cref{fig:trace:like:ex2}. We notice that both C-BART and A-C-BART are able to move towards the true likelihood region, suggesting the efficiency of our method in exploring the parameter space.
	
	\begin{figure}
		\centering
		\includegraphics[width = 0.8\linewidth]{"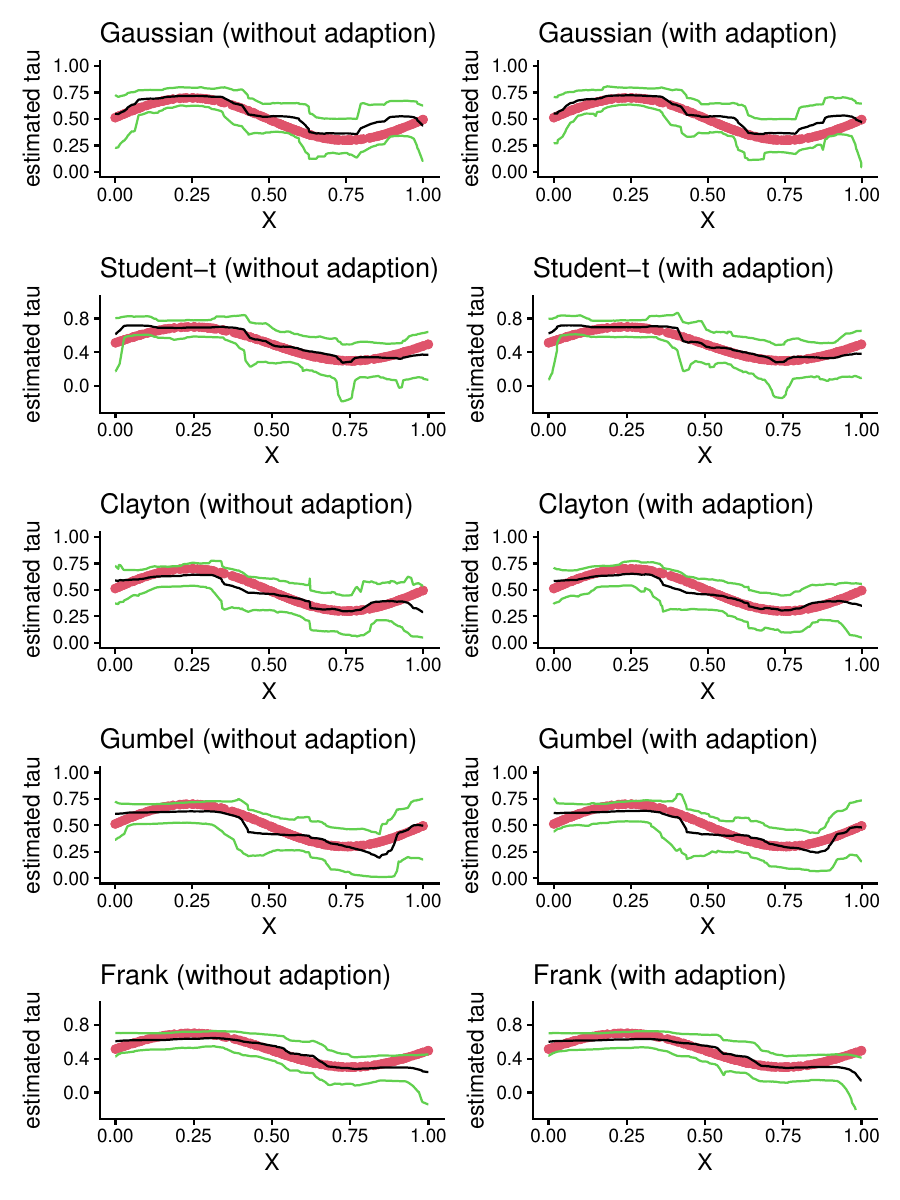"}
		\caption{Comparison of BART fit obtained from our analyses with synthetic datasets generated using a non-linear Kendall's tau, for the first replicate. The plots are obtained by running 4 parallel chains, each one with 3000 MCMC iterations. The left column shows analyses with C-BART and the right column shows analyses with A-C-BART. The green lines denote the 95\% credible interval; the black line denotes the posterior estimate of Kendall's tau; and the red points show the true values of conditional Kendall's tau.}
		\label{fig:trace:pred:ex2}
	\end{figure}
	
	\begin{figure}
		\centering
		\includegraphics[width = 0.8\linewidth]{"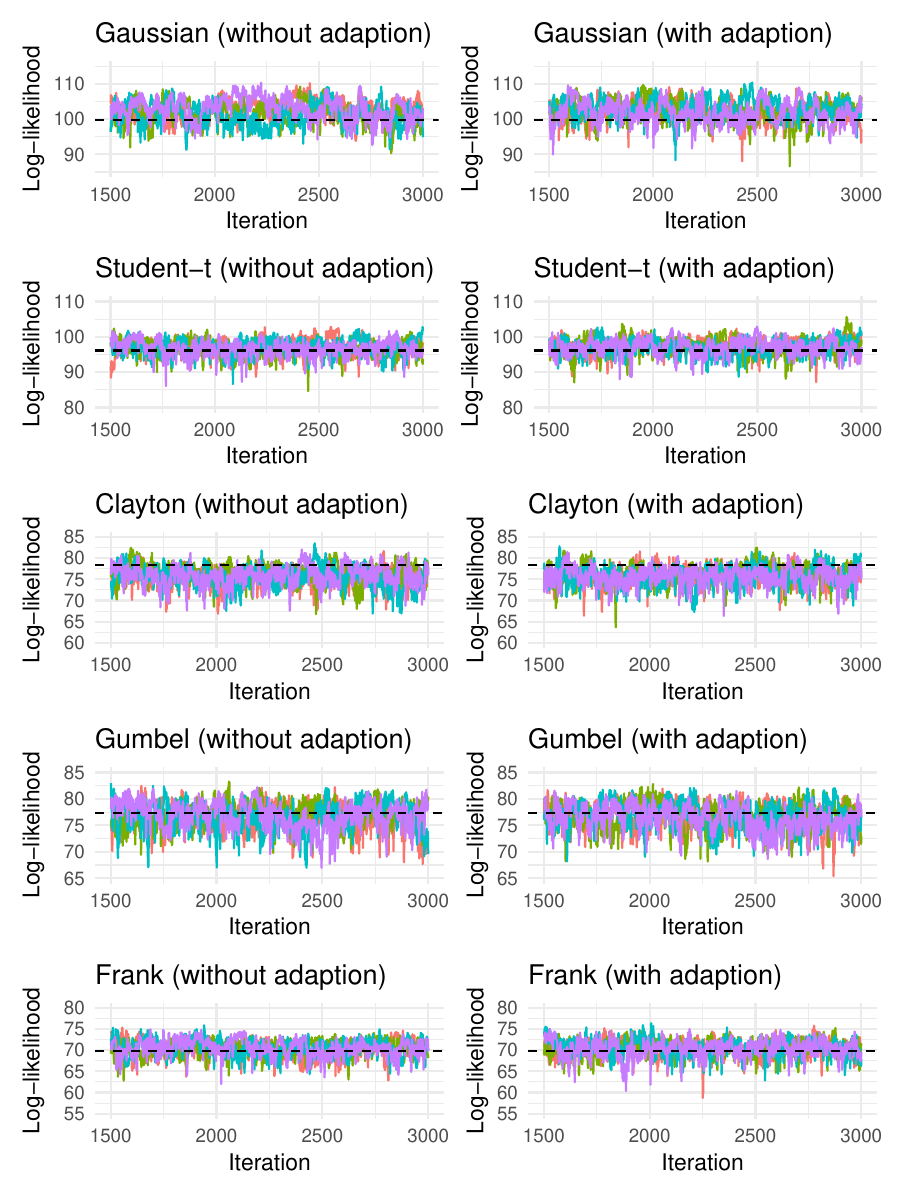"}
		\caption{Trace plots of the log-likelihood obtained from our analyses with synthetic datasets generated using a non-linear Kendall's tau, for the first replicate. The plots are obtained by running 4 parallel chains, each one with 3000 MCMC iterations. The left column shows analyses with C-BART and the right column shows analyses with A-C-BART. The horizontal dashed black line represents the true log-likelihood of the data-generating tree.}
		\label{fig:trace:like:ex2}
	\end{figure}
	
	\section{CIA world fact data}\label{app:cia}
	Additional figures for simulated copulas
	\begin{figure}
		\centering
		\includegraphics[width = 0.75\linewidth]{"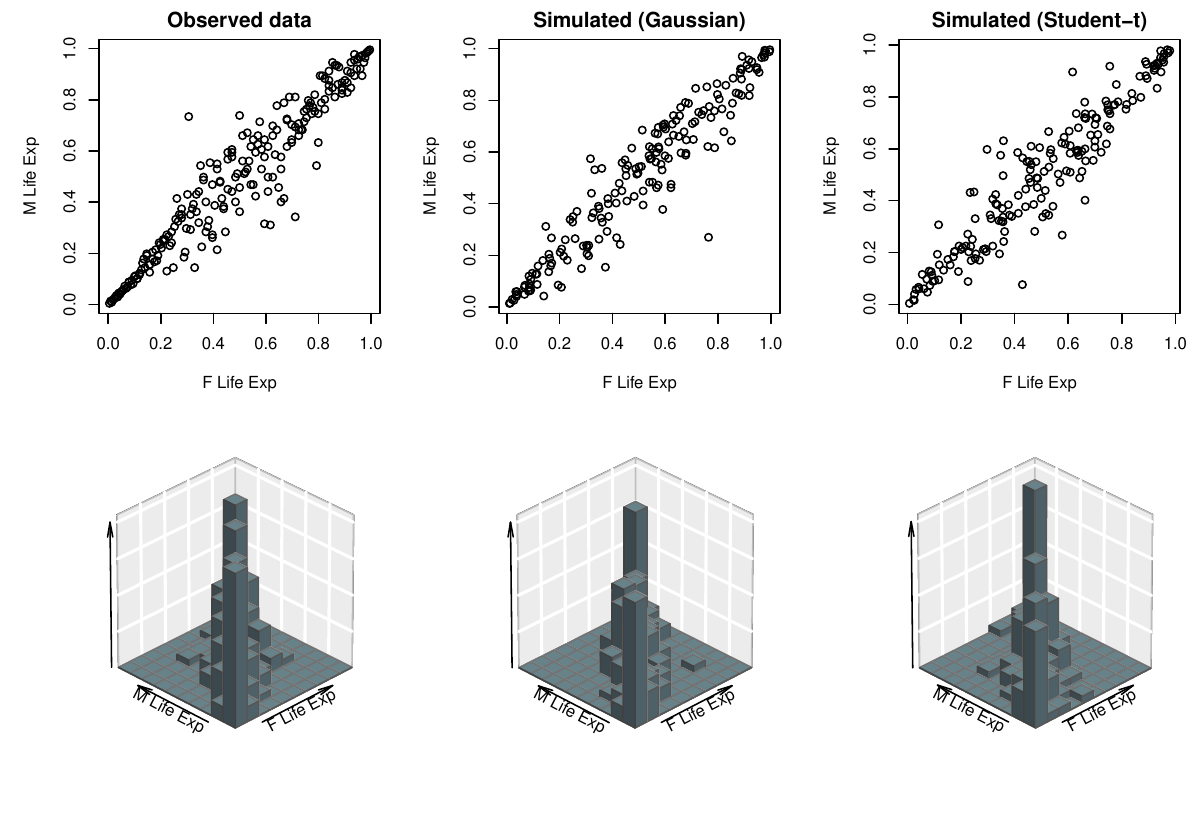"}
		\caption{Scatter plots of pseudo observations and simulated values from the fitted copulas obtained for life expectancy of male and female populations using C-BART with 5 trees.}
		\label{fig:pseudo:LE:woa}
	\end{figure}
	\begin{figure}
		\centering
		\includegraphics[width = 0.75\linewidth]{"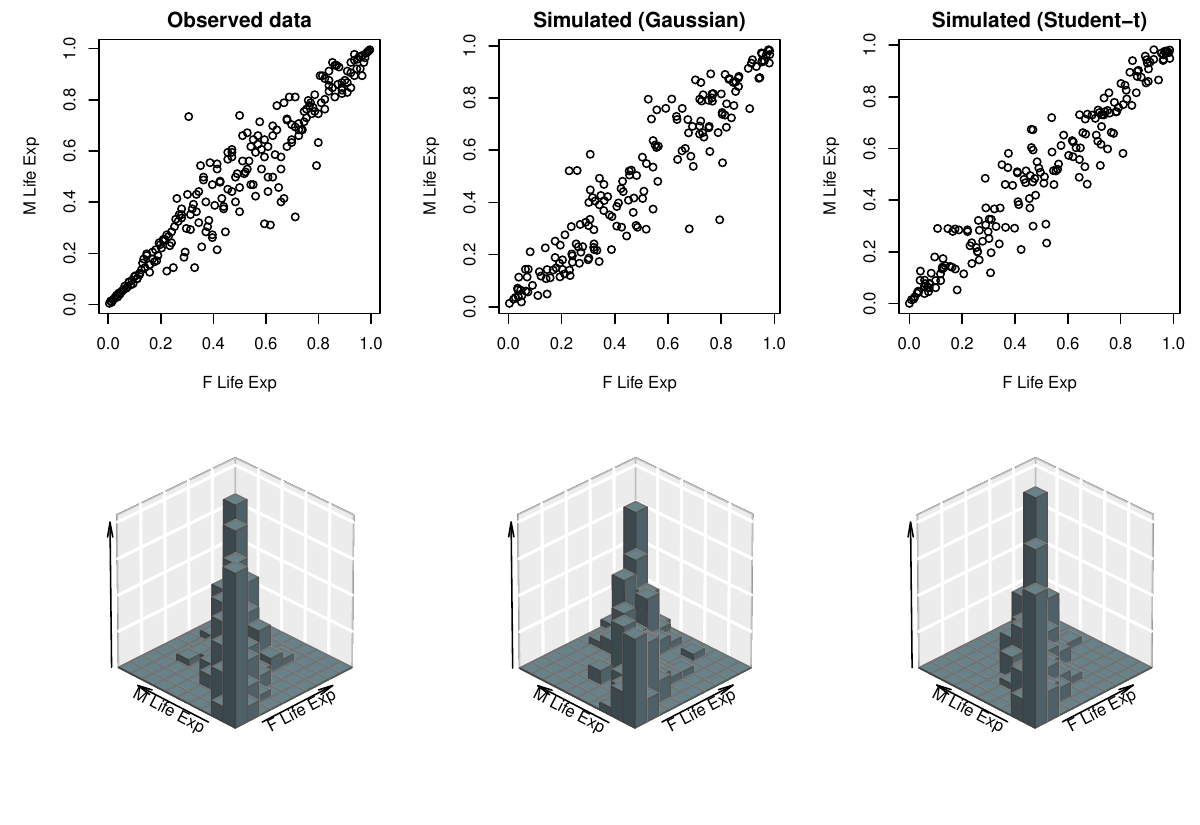"}
		\caption{Scatter plots of pseudo observations and simulated values from the fitted copulas obtained for life expectancy of male and female populations using A-C-BART with 5 trees.}
		\label{fig:pseudo:LE:wa}
	\end{figure}
	
	\begin{figure}
		\centering
		\includegraphics[width = 0.75\linewidth]{"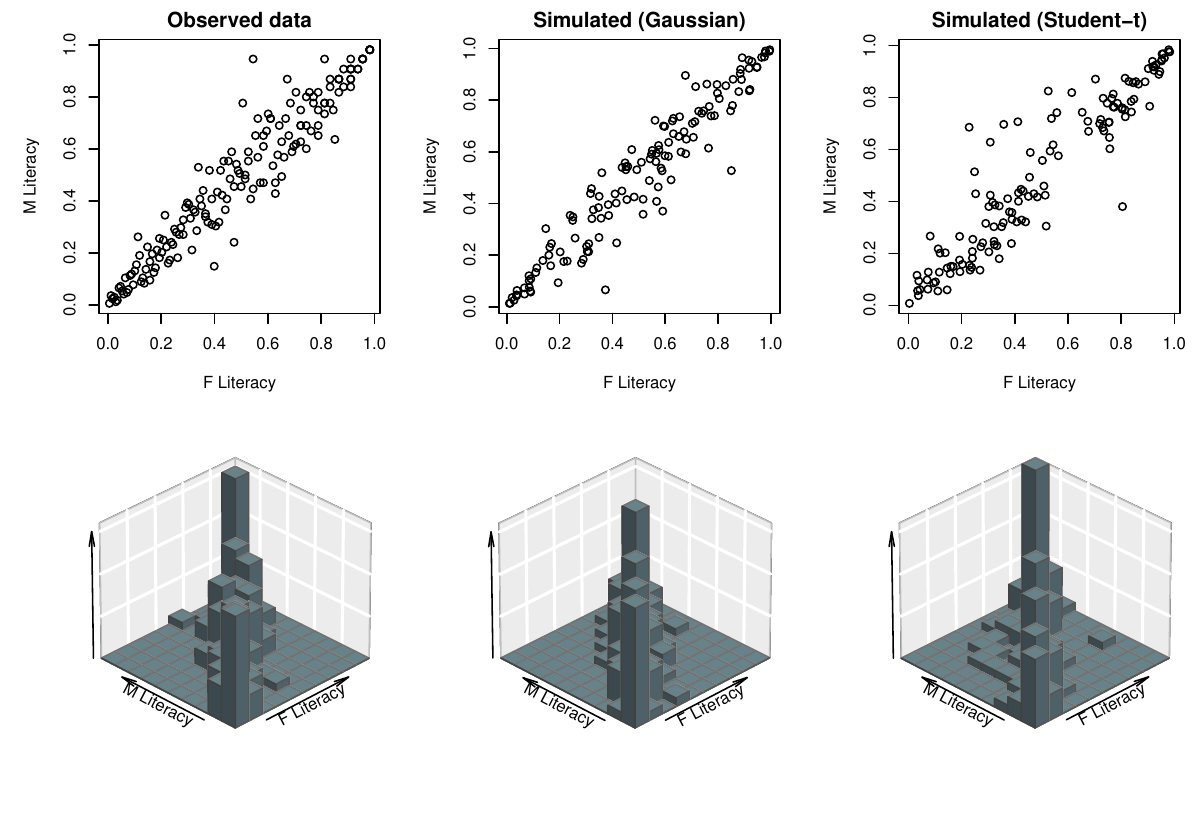"}
		\caption{Scatter plots of pseudo observations and simulated values from the fitted copulas obtained for literacy of male and female population using C-BART with 5 trees.}
		\label{fig:pseudo:LT:woa}
	\end{figure}
	\begin{figure}
		\centering
		\includegraphics[width = 0.75\linewidth]{"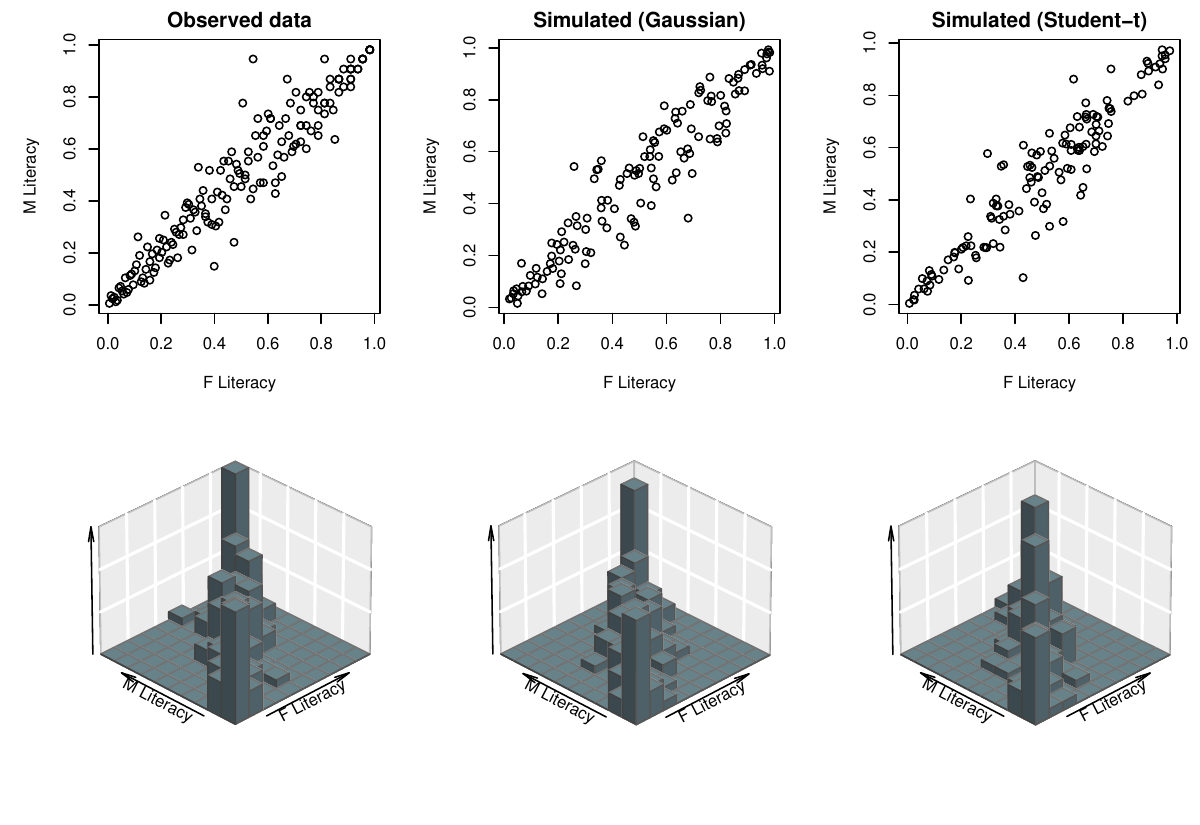"}
		\caption{Scatter plots of pseudo observations and simulated values from the fitted copulas obtained for literacy of male and female population using A-C-BART with 5 trees.}
		\label{fig:pseudo:LT:wa}
	\end{figure}

	Additional figures related to our analyses with 10 trees.
	\begin{figure}[H]
		\centering
		\includegraphics[width = 0.75\linewidth]{"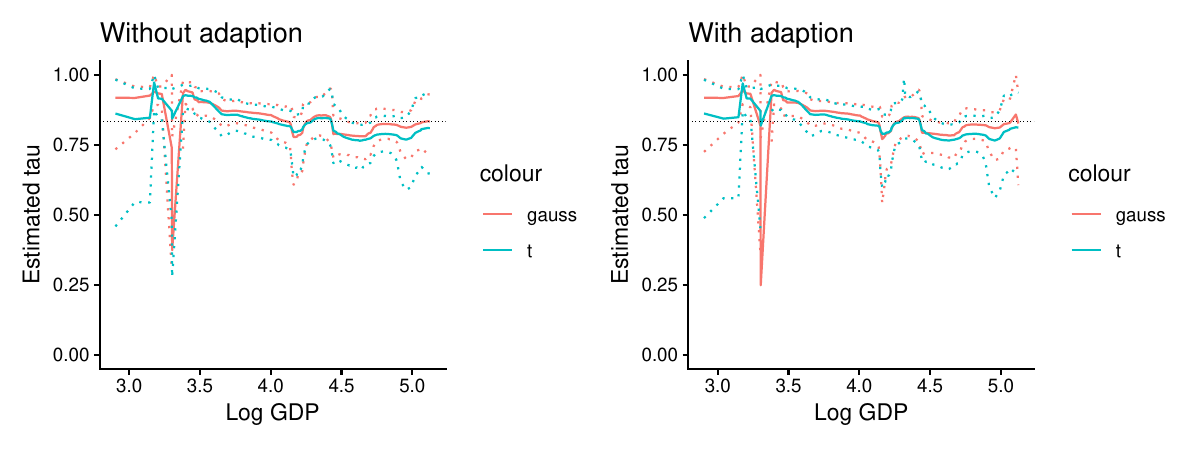"}
		\caption{Estimated dependence between male life expectancy and female life expectancy conditional on log-GDP. For modelling, we use 10 trees and run 4 parallel chains of 50000 iterations. For posterior inference, we discard the first 5000 samples.}
		\label{fig:taus:LE10}
	\end{figure}
	
	\begin{figure}[H]
		\centering
		\includegraphics[width = 0.75\linewidth]{"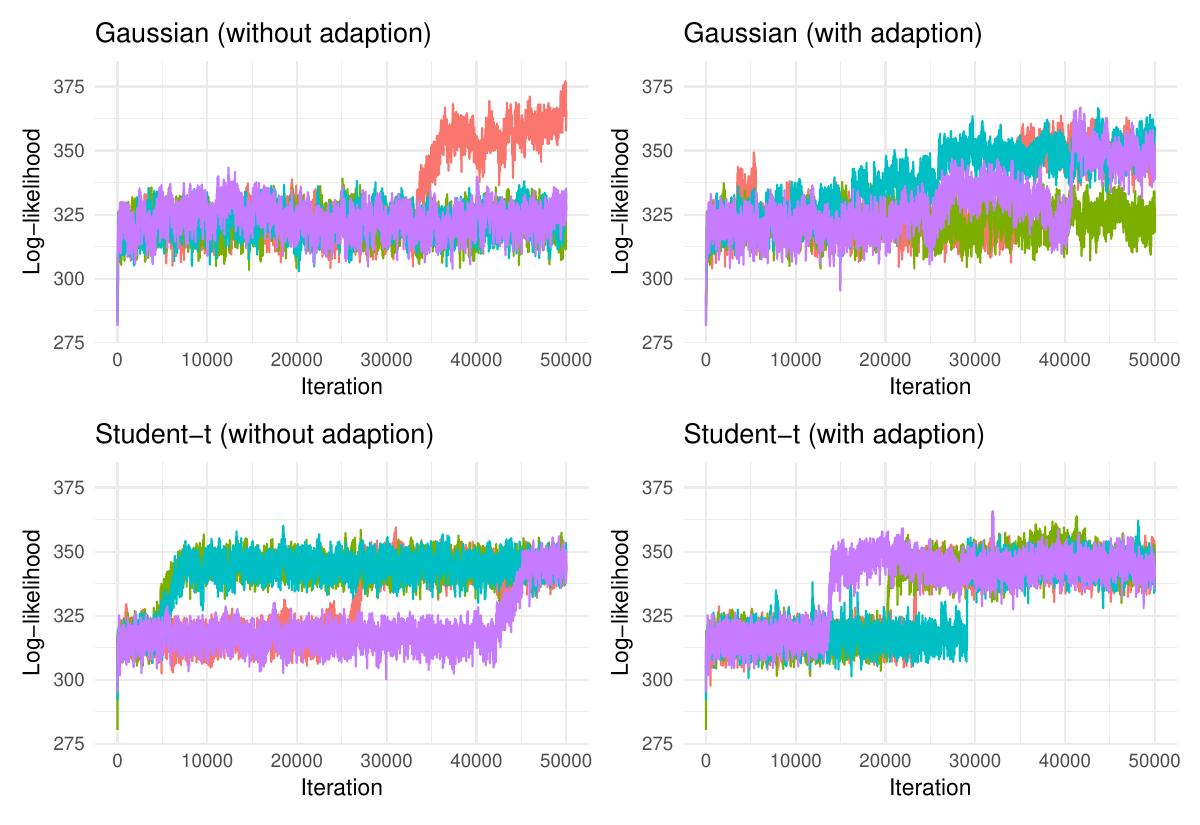"}
		\caption{Trace plots of the log-likelihood obtained from our analyses with life expectancies of the male and female populations. The plots are obtained by running 4 parallel chains, each with 50000 MCMC iterations and 10 trees. The left columns shows analyses with C-BART and the right column shows analyses with A-C-BART.}
		\label{fig:trace:like:real:LE10}
	\end{figure}
	
	\begin{figure}[H]
		\centering
		\includegraphics[width = 0.75\linewidth]{"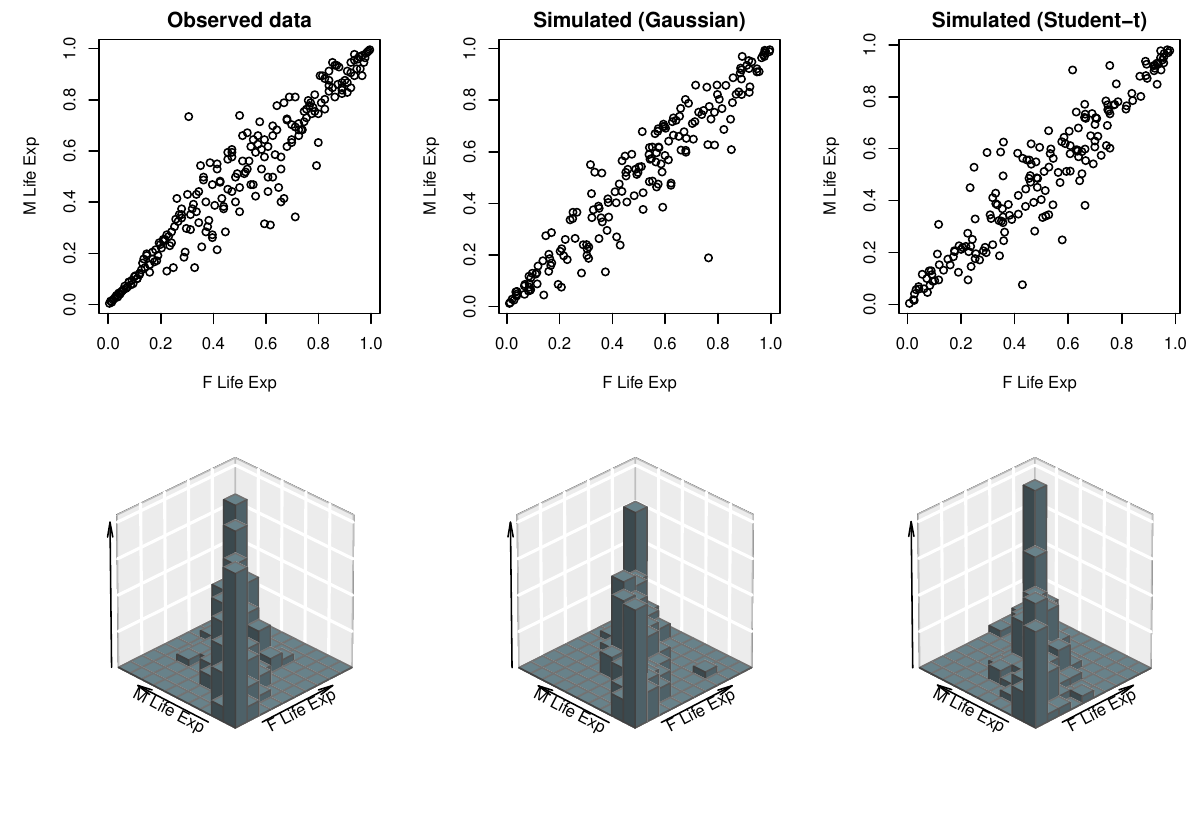"}
		\caption{Scatter plots of pseudo observations and simulated values from the fitted copulas obtained for life expectancy of male and female population using C-BART with 10 trees.}
		\label{fig:pseudo:LE:woa10}
	\end{figure}
	
	\begin{figure}[H]
		\centering
		\includegraphics[width = 0.75\linewidth]{"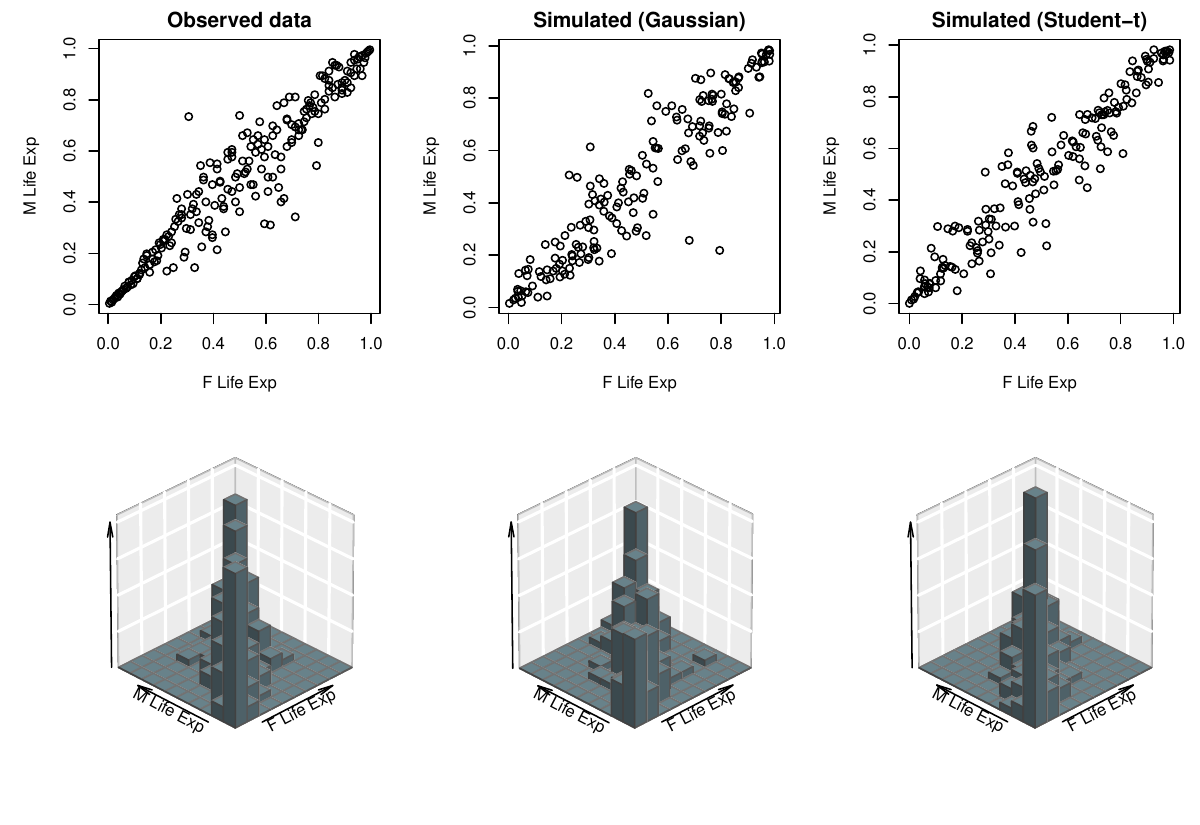"}
		\caption{Scatter plots of pseudo observations and simulated values from the fitted copulas obtained for life expectancy of male and female population using A-C-BART with 10 trees.}
		\label{fig:pseudo:LE:wa10}
	\end{figure}
	
	\begin{figure}[H]
		\centering
		\includegraphics[width = 0.75\linewidth]{"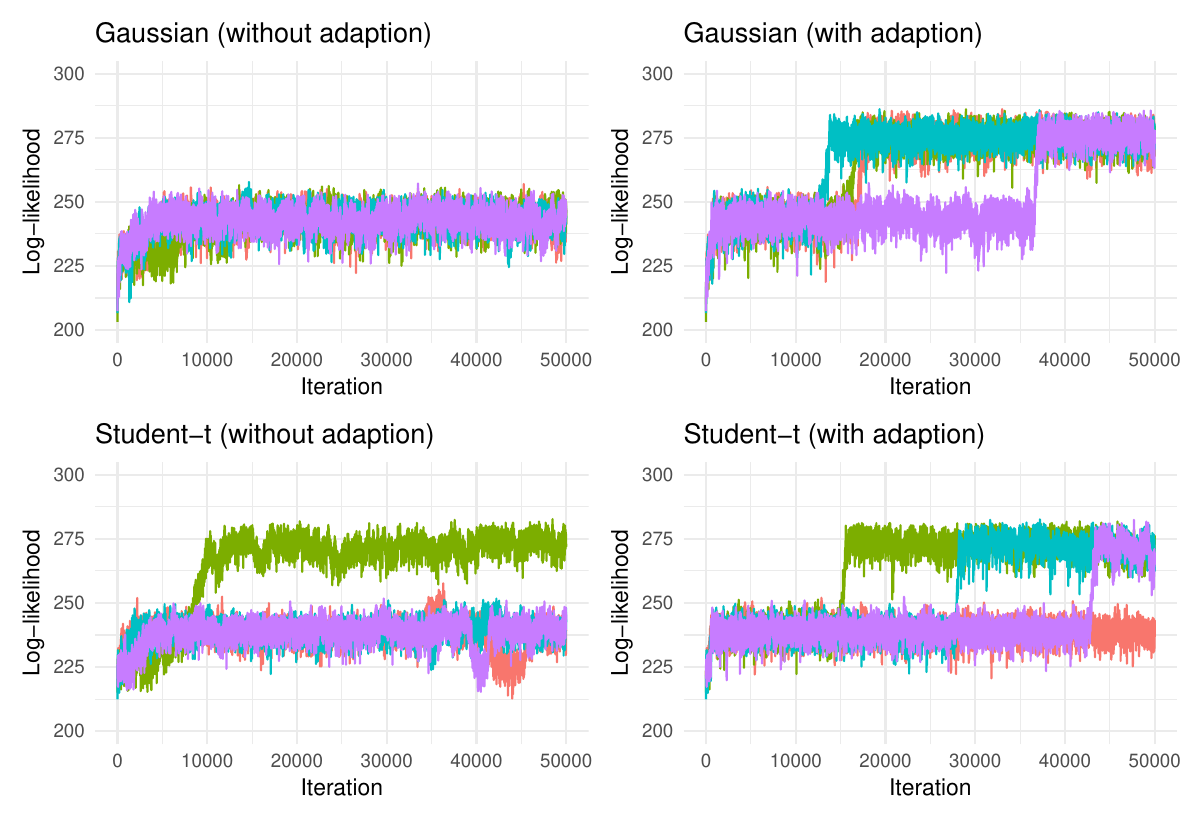"}
		\caption{Trace plots of the log-likelihood obtained from our analyses with literacy rates of the male and female populations. The plots are obtained by running 4 parallel chains, each with 50000 MCMC iterations and 10 trees. The left columns shows analyses with C-BART and the right column shows analyses with A-C-BART.}
		\label{fig:trace:like:real:LT10}
	\end{figure}
	
	\begin{figure}[H]
		\centering
		\includegraphics[width = 0.75\linewidth]{"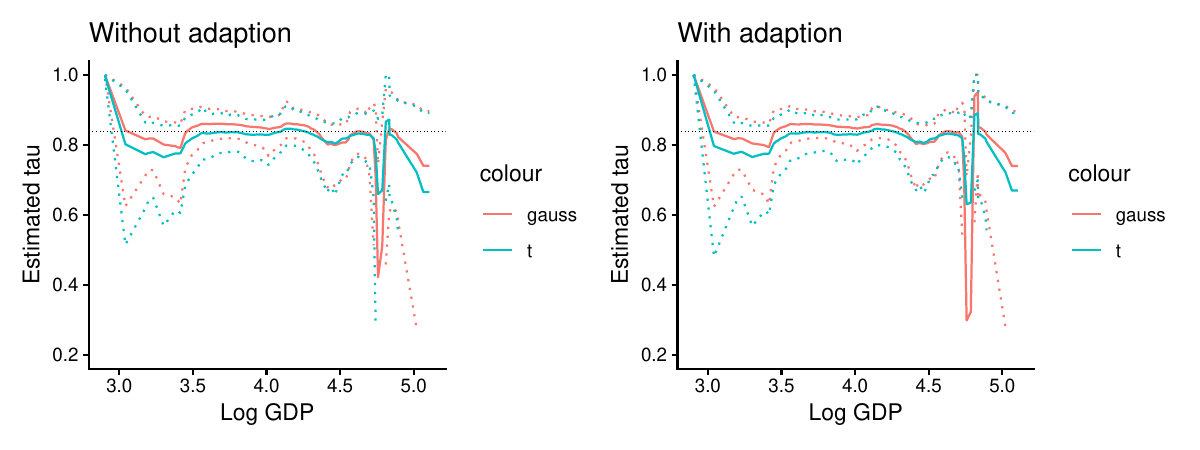"}
		\caption{Estimated dependence between male literacy and female literacy, conditional on log-GDP. For modelling, we use 10 trees and run 4 parallel chains of 50000 iterations. For posterior inference, we discard the first 5000 samples.}
		\label{fig:taus:LT10}
	\end{figure}
	\begin{figure}[H]
		\centering
		\includegraphics[width = 0.75\linewidth]{"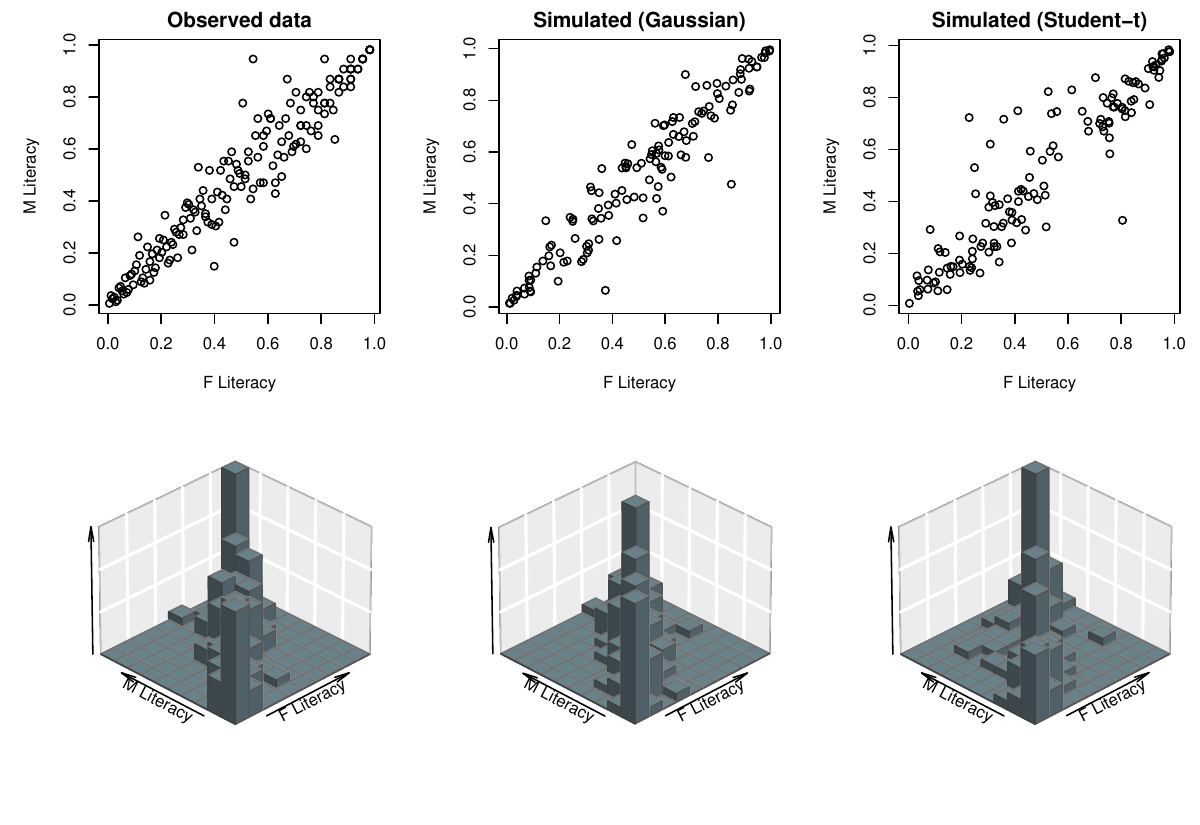"}
		\caption{Scatter plots of pseudo observations and simulated values from the fitted copulas obtained for literacy of male and female population using C-BART with 10 trees.}
		\label{fig:pseudo:LT:woa10}
	\end{figure}
	
	\begin{figure}[H]
		\centering
		\includegraphics[width = 0.75\linewidth]{"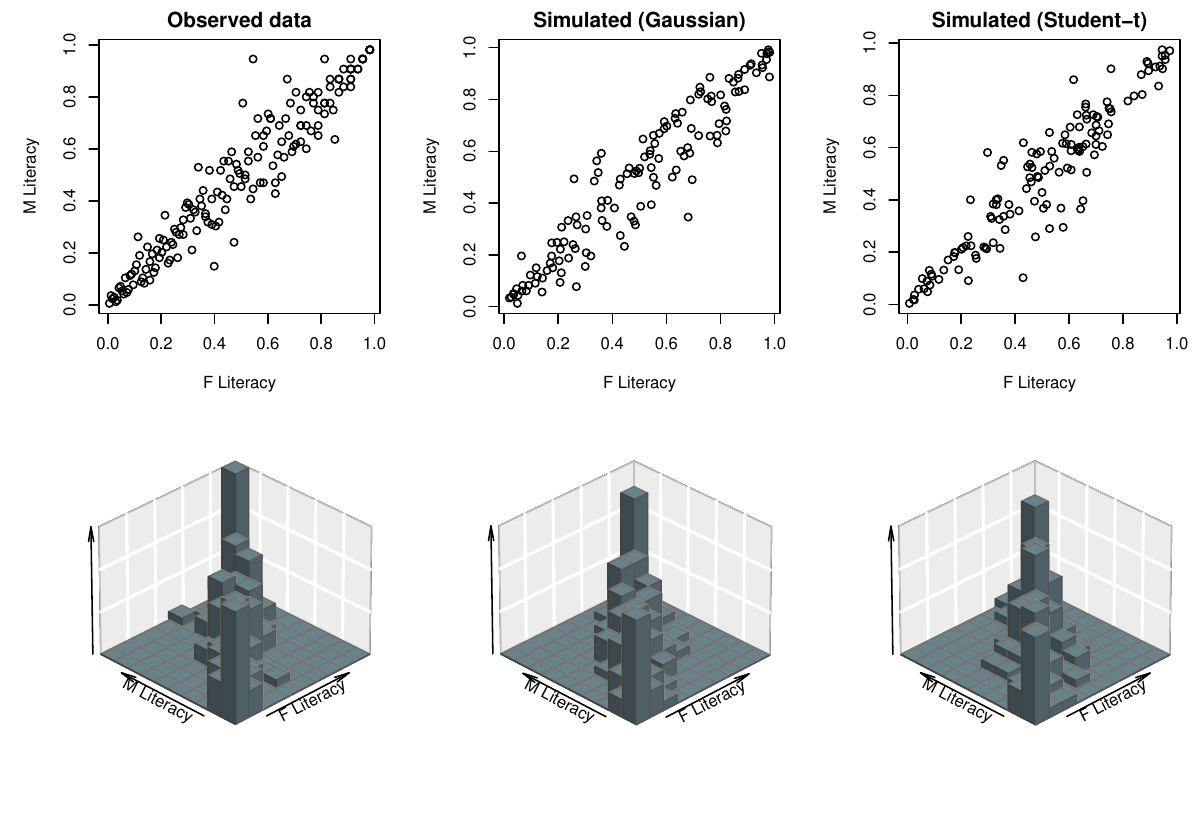"}
		\caption{Scatter plots of pseudo observations and simulated values from the fitted copulas obtained for literacy of male and female population using A-C-BART with 10 trees.}
		\label{fig:pseudo:LT:wa10}
	\end{figure}

\end{document}